\newif\ifarxiv
\arxivtrue

\ifarxiv
\documentclass[11pt,a4paper]{article}
\else
\documentclass[acmsmall,screen,review]{acmart}
\acmJournal{TEAC}
\fi

\ifarxiv
\usepackage[pagebackref]{hyperref}
\usepackage[small]{caption}
\usepackage[sort,numbers]{natbib}
\usepackage{authblk}
\usepackage{fullpage}
\usepackage{amssymb}
\else
\usepackage{hyperref}
\usepackage{caption}
\theoremstyle{acmplain}
\fi

\makeatletter
\AtBeginDocument{%
  \def\doi#1{\href{https://doi.org/#1}{\nolinkurl{doi:#1}}}%
 }
\makeatother

\usepackage[utf8]{inputenc}
\usepackage{url}
\usepackage{graphicx}
\usepackage{amsmath}
\usepackage{mathtools}
\usepackage{amsthm}
\usepackage{microtype}
\usepackage{booktabs}
\usepackage{algorithm}
\usepackage{algorithmic}
\usepackage[capitalize]{cleveref}
\usepackage{tabularx,booktabs,multirow}
\usepackage{caption}
\usepackage{subcaption}
\usepackage{thm-restate}

\usepackage{pifont}
\usepackage{paralist}

\newcommand{\cmark}{{\large \ding{51}}}
\newcommand{\xmark}{{\large \ding{55}}}

\newtheorem{theorem}{Theorem}
\newtheorem{corollary}[theorem]{Corollary}

\newtheorem{observation}[theorem]{Observation}
\newtheorem{definition}[theorem]{Definition}

\newtheorem{lemma}[theorem]{Lemma}

\crefname{theorem}{Theorem}{Theorems}
\crefname{corollary}{Corollary}{Corollaries}
\crefname{observation}{Observation}{Observations}
\crefname{lemma}{Lemma}{Lemmas}
\crefname{table}{Table}{Tables}
\crefname{figure}{Figure}{Figures}

\usepackage{xspace}

\newcommand{\td}{{{\mathrm{td}}}}
\newcommand{\at}{\mathrm{td}}
\newcommand{\Diff}{{{\mathrm{Diff}}}}
\newcommand{\dDi}{{{\Delta\mathrm{Diff}}}}
\newcommand{\Vor}{{{\mathrm{Vor}}}}
\newcommand{\dVor}{{{\Delta\mathrm{Vor}}}}
\newcommand{\diDiff}{temporal difference diffusion\xspace}
\newcommand{\dVoronoi}{temporal difference Voronoi}
\newcommand{\nc}{nice\xspace}

\newcommand{\Pp}{\mathcal{P}}
\newcommand{\Gg}{\mathcal{G}}
\newcommand{\Cc}{\mathcal{C}}
\newcommand{\Tt}{\mathcal{T}}
\newcommand{\Bb}{\mathcal{B}}

\newcommand{\J}{ \jmath} 

\newcommand{\appear}[1]{\tau_{#1}}
\newcommand{\jointime}{{\bar{t}}}

\newcommand{\bigO}{\mathcal{O}}
\newcommand{\NN}{\mathbb{N}}
\newcommand{\ZZ}{\mathbb{Z}}

\newcommand{\footprint}[1]{#1_\downarrow}

\newcommand{\forest}{\mathcal{F}}
\newcommand{\stopForest}{\overline{\mathcal{F}}}

\DeclarePairedDelimiter{\floor}{\lfloor}{\rfloor}
\DeclarePairedDelimiter{\ceil}{\lceil}{\rceil}
\DeclarePairedDelimiter{\abs}{\lvert}{\rvert}

\newcommand{\mymod}[1]{\underline{#1}}

\newcommand{\lreach}{\omega_<}
\newcommand{\rreach}{\omega_>}
\newcommand{\reach}{\omega}

\DeclareMathOperator{\med}{med}

\usepackage{algorithm}
\usepackage{capt-of}

\usepackage{tikz}
\usetikzlibrary{arrows,decorations.pathmorphing,decorations.pathreplacing,backgrounds,positioning,fit,matrix,shapes,backgrounds,calc}
\usetikzlibrary{shapes,calc}
\tikzset{
	vertex/.style={
		circle,
		scale=0.5,
		draw=black,
		fill = gray!7,
		thick,
	},
	vertexP/.style={
		vertex,
		fill=\colorP,
	},
	vertexQ/.style={
		vertex,
		fill=\colorQ,
	},
	vertexG/.style={
		vertex,
		fill=\colorG,
	},
	edge/.style={
		draw=black,
		very thick,
	},
	vboundary/.style={ 
		rectangle,
		draw,
		minimum width=1em,
		minimum height=0pt,
		inner sep=0pt
	},
	reach/.style={
		draw=#1,
		fill=#1,
		fill opacity=0.2,
	},
	timeline/.style={
		draw=black,
		semithick,
		>=stealth,
	},
	timelabel/.style={
		font=\small,
	},
	group/.style={
		draw,
		thin,
		rounded corners,
	},
	blob/.style={
		draw,
		thin,
		rounded corners=5mm,
	},
	path/.style={
		draw,
		thick,
		decorate,
		decoration={snake, segment length=7mm, amplitude=.3mm},
	},
}

\def \colorP {blue!60}
\def \colorQ {red!60}
\def \colorG {gray!90}

\def\timelabelx{-1}
\def\starttime{0}
\def\vertexlabely{-1.2}
\newcommand{\gamegrid}[2]{
	\def\n{#1}
	\def\T{#2}
	\draw
		\foreach \time in {\starttime,...,\T} {
			(0, {\T-\time}) node {\time}
			\foreach \i in {1,...,\n} { 
				(\i, {\T-\time}) node[vertex] (v-\i-\time) {}
			}
		}
	;
	\node[rotate=90] at  (\timelabelx,{(\T-\starttime)/2}){time};
	\node[anchor=base] at ({\n/2+0.5}, \vertexlabely-0.8) {vertex};
}

\newcommand{\colorvertex}[3]{
	\draw \foreach \time in {#2,...,\T}{
		(v-#1-\time) node[#3] {}
	};
}

\newcommand{\ints}[2]{[#1, #2]}
\newcommand{\zeroto}[1]{\ints{0}{#1}}
\newcommand{\oneto}[1]{[#1]}

\newcommand{\case}[2]{\vspace{-.9\baselineskip}\paragraph{Case~#1: #2}}

\ifarxiv
\title{Two Influence Maximization Games on Graphs Made Temporal\thanks{An extended abstract of this article appeared in the proceedings of the Thirtieth International Joint Conference on Artificial
		Intelligence (IJCAI 2021), pages 45--51, 2021.}}

\author{Niclas Boehmer\footnote{Supported by DFG projects MaMu (NI~369/19) and ComSoc-MPMS (NI 369/22).}}
\author{Vincent Froese}
\author{Julia Henkel} 
\author{Yvonne Lasars} 
\author{Rolf Niedermeier} 
\author{Malte Renken\footnote{Supported by DFG project MATE (NI~369/17).}}

\affil{\small
	Technische Universit\"at Berlin, Faculty~IV, Algorithmics and Computational 
	Complexity\protect\\
	\{niclas.boehmer, vincent.froese\}@tu-berlin.de,\protect\\
	\{henkel, y.lasars\}@campus.tu-berlin.de,\protect\\
	m.renken@tu-berlin.de}

\date{\today}

\else

	\title[Two Influence Maximization Games on Graphs Made Temporal]{{Two Influence Maximization Games on Graphs Made Temporal}}

\author[N. Boehmer]{Niclas~Boehmer}
\orcid{}
\affiliation{
	\institution{Technische Universit\"at Berlin}
	\department{Algorithmics and Computational Complexity}
	\streetaddress{Ernst-Reuter Platz 7}
	\postcode{10587}
	\city{Berlin}
	\country{Germany}
}
\email{niclas.boehmer@tu-berlin.de}

\author[V. Froese]{Vincent Froese}
\orcid{}
\affiliation{
	\institution{Technische Universit\"at Berlin}
	\department{Algorithmics and Computational Complexity}
	\streetaddress{Ernst-Reuter Platz 7}
	\postcode{10587}
	\city{Berlin}
	\country{Germany}
}
\email{vincent.froese@tu-berlin.de}

\author[J. Henkel]{Julia Henkel}
\orcid{}
\affiliation{
	\institution{Technische Universit\"at Berlin}
	\department{Algorithmics and Computational Complexity}
	\streetaddress{Ernst-Reuter Platz 7}
	\postcode{10587}
	\city{Berlin}
	\country{Germany}
}
\email{henkel@campus.tu-berlin.de}

\author[Y. Lasars]{Yvonne Lasars}
\orcid{}
\affiliation{
	\institution{Technische Universit\"at Berlin}
	\department{Algorithmics and Computational Complexity}
	\streetaddress{Ernst-Reuter Platz 7}
	\postcode{10587}
	\city{Berlin}
	\country{Germany}
}
\email{y.lasars@campus.tu-berlin.de}

\author[R. Niedermeier]{Rolf Niedermeier}
\orcid{}
\affiliation{
	\institution{Technische Universit\"at Berlin}
	\department{Algorithmics and Computational Complexity}
	\streetaddress{Ernst-Reuter Platz 7}
	\postcode{10587}
	\city{Berlin}
	\country{Germany}
}

\author[M. Renken]{Malte Renken}
\orcid{}
\affiliation{
	\institution{Technische Universit\"at Berlin}
	\department{Algorithmics and Computational Complexity}
	\streetaddress{Ernst-Reuter Platz 7}
	\postcode{10587}
	\city{Berlin}
	\country{Germany}
}
\email{m.renken@tu-berlin.de}
\fi

\begin{document}

\ifarxiv
\maketitle

\begin{abstract}
 To address the dynamic nature of real-world 
 networks,
 we generalize competitive diffusion games and Voronoi games from static
 to temporal graphs, where edges may appear or disappear over time. 
This establishes a new direction of studies in the 
 area of graph games, motivated by applications such as influence spreading. 
 As a first step, we investigate the existence of Nash equilibria in 2-player competitive 
diffusion and
 Voronoi games on different temporal graph classes.
Even when restricting our studies to temporal trees and cycles, this turns out to be a challenging
 undertaking, revealing significant differences between the two games in the temporal setting.
 Notably, both games are equivalent on static trees and cycles.
 Our two main technical results are (algorithmic) proofs for the 
 existence of Nash equilibria in 2-player competitive diffusion and temporal 
 Voronoi games when the edges are restricted not to disappear over time.

 \medskip

\end{abstract} 
\fi

\section{Introduction}
As graph games help us to reason about a networked world, playing games on graphs is an intensively researched topic since decades. 
In this work, we focus on competitive games on undirected graphs.
Here, some external parties influence a small subset of
agents who then spread some information through the network.
Typical application scenarios for these occur, for instance, when
political parties want to gain influence in a social network or in the 
context of viral marketing.

Looking at two prominent, somewhat similar representatives, 
namely competitive diffusion games and 
Voronoi games, we put forward to model network dynamics more realistically. 
Specifically, while to the best of our knowledge almost all 
work on graph games
focused on \emph{static} graphs, we initiate the study 
of these two games on \emph{temporal} graphs. Roughly speaking, 
in a temporal graph, the edge set may evolve 
over discrete time steps, while the vertex set remains unchanged, yielding 
time-ordered graph layers with different edge sets.
Moving to the temporal setting has dramatic consequences. For example, while  
competitive 
diffusion games and Voronoi games are equivalent on static trees and 
cycles~\cite{DBLP:journals/tcs/SunSXZ20}
and we understand well their properties in these simple but important special cases, they are no longer equivalent in the temporal case and their properties are much more challenging to analyze.

Our study takes a first step towards understanding both games on temporal paths, trees, and cycles.
It turns out that these backbone structures of graphs already 
confront us with several technically challenging questions when looking for the existence 
(and computation) of Nash equilibria---one of the most fundamental game-theoretic concepts.
We refer to the next section for all formal definitions and examples of the two (temporal) games.
Intuitively speaking, in both games one can think of 
each player having a color and trying to color as many vertices as 
possible
by her own color; the coloring process starts in a vertex freely chosen by each player 
and acts through the neighborhood relation of the graph.
Herein, the distance of a vertex to the start vertices plays a central role.
In both games, a player colors all vertices that are closest  
to her start vertex. Moreover, while in competitive diffusion games a vertex 
that is at the same 
distance to two start vertices of competing players can still get one of the 
two colors, this is not the case for Voronoi games.  

\paragraph{Related work.}
Competitive diffusion games were introduced
by \citet{alon}.
Research on competitive diffusion games so far mainly focused on 
the existence of Nash equilibria on a variety of graph 
classes for different numbers of 
players~\cite{alon,bult,fuku,roshan,small,suke,take}.
Also, the (parameterized) computational complexity of deciding the existence of 
a Nash equilibrium has been studied~\cite{etes,ito}.

Voronoi games have been originally studied for 
a one-dimensional
or two-dimensional continuous space
\cite{DBLP:journals/tcs/AhnCCGO04,DBLP:journals/jco/BanikBD13,DBLP:conf/isaac/BergKM19,DBLP:journals/dcg/CheongHLM04,DBLP:journals/comgeo/FeketeM05}.
There, it is typically assumed that players choose their initial sets of
points sequentially and that a player wins the game if a certain fraction of all
points is closest to her.
Voronoi games on graphs have also been studied on 
different graph classes for various numbers of 
players~\cite{DBLP:journals/tcs/BandyapadhyayBDS15,DBLP:conf/esa/DurrT07,DBLP:conf/wine/FeldmannMM09,DBLP:conf/mfcs/MavronicolasMPS08,DBLP:journals/tcs/SunSXZ20,DBLP:journals/jgaa/TeramotoDU11}.
Again a focus lies on determining for which graphs a Nash equilibrium 
exists 
and how to compute one. 

From a broader perspective, analyzing games played on graphs from a game-theoretic 
perspective is an intensively researched topic; recent examples
include Schelling games 
\cite{ChauhanLM18,DBLP:conf/ijcai/ElkindGISV19}, 
$b$-matching games~\cite{DBLP:conf/ijcai/KumabeM20}, or network creation 
games~\cite{DBLP:conf/ijcai/Echzell0LM20}.

As mentioned before, we know of basically no work 
systematically studying
graph games in a temporal setting.
The only exception we are aware of is in the context of 
pursuit-evasion games:
\citet{erlebach} studied the pursuit-evasion game of cops 
and robbers 
on some specific
temporal graphs, namely so-called edge-periodic graphs. 
Also different from our studies, their focus was on computing 
winning strategies for the players. 
\citet{DBLP:conf/mfcs/MorawietzRW20} and 
\citet{MW21} extended this study, 
also answering an open question of \citet{erlebach}.

\paragraph{Our contributions.}
We put forward the study of game-theoretic models on
\emph{temporal} graphs.
We do so by generalizing two well-studied (static) graph games to temporal graphs, namely competitive diffusion games and Voronoi games.
For the two resulting temporal graph games we analyze 
the (constructive) existence 
of Nash equilibria 
on different temporal graph classes, focusing 
on different types of temporal paths, trees, and cycles (see \Cref{fi:ov}
for an overview of our results).
We observe that, in contrast to the static case where both games are 
equivalent and a Nash equilibrium is guaranteed to exist~\cite{roshan}, the games exhibit far 
more complex dynamics and a quite different behavior on temporal trees and 
cycles.
Our main results are two involved proofs of guaranteed existence of Nash equilibria, 
namely in temporal diffusion games and temporal Voronoi games on
so-called monotonically growing temporal cycles. 
On a high level, one conclusion from our work is that temporal Voronoi games seem to exhibit a more intricate behavior than temporal diffusion games. 

\begin{table}[t]
	\centering
	\footnotesize
        \def\arraystretch{1.3}
        \caption{Overview of our results. ``\xmark''~means that a Nash equilibrium 
          is 
          not guaranteed to exist. ``\cmark''~means that a Nash equilibrium always 
          exists. 
          Entries in parentheses are implied by other table cells.
          See \Cref{se:prel} for formal definitions. Voronoi games and diffusion games  
          on static trees and static cycles are guaranteed to admit a Nash equilibrium.
        }
	\begin{tabular}{l llll}
	\toprule
	&	 Temporally 			
			& Monotonically & Monotonically\\
	& 	connected	& growing  & shrinking\\		
    \midrule
    \textbf{Diffusion} \\
	Temporal Paths 	& \hspace{-4pt}(\cmark)
	& \hspace{-4pt}(\cmark) &  \xmark~(Theorem 
\ref{thm:shrink}) \\
	Temporal Trees & \cmark~(Theorem \ref{thm:diff-tree})& 
	\hspace{-4pt}(\cmark) & \hspace{-4pt}(\xmark) \\
	Temporal Cycles		& \xmark~(Theorem
	\ref{sup_cycle}) 	& 
	\cmark~(Theorem \ref{mono_cycle}) &  \xmark~(Theorem 
\ref{thm:shrinkc})\\
	 \midrule
	 \midrule
    \textbf{Voronoi}\\
	Temporal Paths 		& 
\xmark~(Theorem \ref{voronoi_nonex}) & \hspace{-4pt}(\cmark) & 
\xmark~(Corollary  
\ref{thm:shrinkVor})\\
	Temporal Trees   &
	\hspace{-4pt}(\xmark) & \cmark~(Theorem \ref{thm:voronoi-tree}) & 
\hspace{-4pt}(\xmark)\\
	Temporal Cycles		& 
\xmark~(Theorem \ref{voronoi_nonex})	& \cmark~(Theorem~\ref{thm:vor-ne-cycle}) &\xmark~(Corollary 
\ref{thm:shrinkVor})\\
	\bottomrule
	\end{tabular}
	\label{fi:ov}
\end{table}

\paragraph{Organization of the paper.} In \Cref{se:prel}, we provide some 
background on temporal graphs and define temporal competitive diffusion and 
temporal Voronoi games. In \Cref{se:TDG}, we analyze temporal competitive 
diffusion games. We first consider these games on 
temporal trees (\cref{se:TDG_tree}) and afterwards on temporal cycles 
(\Cref{se:TDG_cycle}). In \Cref{subsub:moncycles}, we present our technically 
most involved result for diffusion games, i.e., we show that every temporal competitive diffusion 
game on a 
monotonically growing temporal cycle admits a Nash equilibrium. After that, in 
\Cref{se:vor}, we turn to temporal Voronoi games. Here, the structure is a bit 
different. We first show negative results on monotonically shrinking paths and cycles 
(\Cref{se:vor1}) and then  on
temporally connected  trees and cycles (\Cref{se:vor2}).
Lastly, in \Cref{se:vor3}, we study monotonically growing 
trees and in \Cref{se:vor4}, present our most complicated result showing that a Nash equilibrium exists in all temporal Voronoi games on monotonically growing cycles.

\section{Preliminaries} \label{se:prel}
For $a \le b\in \mathbb{N}$, let $[a,b]:=\{a,a+1,\dots, b\}$, 
$[a,b[:=\{a,a+1,\dots, b-1\}$, and $]a,b]:=\{a+1,\dots, b\}$. 
Further, let 
$[n]:=[1,n]$ for~$n
\in\mathbb{N}$.
For any proposition~$P$, the Iverson bracket~$[P]$ evaluates to~$1$ if~$P$~is true and to~$0$ if~$P$~is false.

\subsection{Temporal graphs}
A \emph{temporal graph}~$\mathcal{G} = (V, (E_t)_{t=1}^\infty)$
consists of a set of vertices~$V$
and a sequence of edge sets~$(E_t)_{t=1}^\infty$ with 
$E_t\subseteq \binom{V}{2}$.
We call the graph~$G_t=(V,E_t)$ the \emph{$t$-th layer} of~$\mathcal{G}$
and $\mathcal{G}_\downarrow = 
(V, E_\downarrow)$ with~$E_\downarrow := \bigcup_{t=1}^\infty E_t$
the \textit{underlying graph} of~$\mathcal{G}$.

If there is an integer~$\tau(\Gg)$
such that $E_t = E_\tau$ for all $t \geq \tau(\Gg)$,
then the minimum such integer~$\tau(\Gg)$ is called the \emph{lifetime} of~$\Gg$.
Intuitively, this means that the graph stops changing after time~$\tau(\Gg)$.
Since the temporal games we consider (which will be defined in \cref{sec:games}) always end after some finite amount of time (depending on the temporal graph),
it does not constitute a loss of generality to assume that all our temporal graphs have finite lifetime~$\tau$.
We will then simply omit specifying $E_i$ for $i > \tau$ (implying that $E_i=E_\tau$ for all $i>\tau$).

A \emph{temporal forest} (resp.\ \emph{tree}, \emph{path}, \emph{cycle}) is a temporal graph
whose underlying graph is a forest (resp.\ path, cycle).
In the case of paths, we will use the convention that the path runs from left to right and that its vertices are denoted $1,\dots, n$ in that order.
Similarly, we will denote the vertices of a cycle by either $1, \dots, n$ or $0,\dots, (n-1)$ in counterclockwise order.

For the definition of temporal Voronoi games, we require a notion of 
temporal distance between two vertices.
In a  temporal graph $\mathcal{G}=(V,(E_t)_{t=1}^\infty)$, we 
define a \emph{temporal walk} from a vertex $v_0$ to a vertex $v_d$ as a 
sequence of tuples 
$(\{v_0,v_1\},t_1),(\{v_1,v_2\},t_2),\dots, (\{v_{d-1},v_d\},t_d)$ such that
the following properties hold:
\begin{compactitem}
  \item $t_i < t_{i+1}$ for all $i\in[d-1]$,
  \item $\{v_{i-1},v_i\}\in E_{t_i}$ for all $i\in[d]$.
\end{compactitem}

We refer to $t_d$ as the \emph{arrival time} of the temporal walk.
Moreover, we call a temporal walk from~$v_0$ to~$v_d$ \emph{foremost} if there 
is no temporal walk from~$v_0$ to~$v_d$ with a smaller arrival time.
We now define the temporal distance $\td(u,v)$ from $u$ to $v$ as the arrival 
time of a foremost walk\footnote{We consider 
foremost walks since earliest arrival seems more natural than other concepts such as \emph{fastest} or \emph{shortest}~\cite{BHNN20} in the context of
influence spreading.
Moreover, we require \emph{strict} inequality between $t_i$~and~$t_{i+1}$ since this follows the static case and the diffusion process more closely.}
from $u$ to~$v$.
If there is no such walk, then set $\td(u, v) = \infty$.
Notably, in contrast to the static case, temporal distances are not necessarily 
symmetric, that is, $\td(u,v)\neq \td(v,u)$ is possible.
By convention, we set~$\td(v, v) = 0$ for any vertex $v$.
For two vertices $u$ and $v$, we say that 
\emph{$u$ reaches $v$ until step $\ell$} if $\td(u,v)\leq \ell$ and that 
\emph{$u$ reaches $v$ in step (or at time) $\ell$} if $\td(u,v)= \ell$.
The set of all vertices reachable from~$v$ until time~$\ell$ in $\Gg$ is denoted by $\Omega_\Gg^\ell(v) := \{w \mid \td(v, w) \leq \ell\}$
where the subscript $\Gg$ is omitted if it is clear from context.

A temporal graph~$\Gg$ is \emph{temporally connected} if $\td(u, v) < \infty$ for all vertex pairs $u, v$.
Further, we call $\Gg$ \emph{monotonically growing} 
if no edge disappears over time, i.e., $E_t\subseteq E_{t+1}$ for all~$t$.
Symmetrically, $\mathcal{G}$~is \emph{monotonically shrinking} if edges 
do not appear over time, i.e., $E_{t+1}\subseteq E_t$ for all~$t$.

\subsection{Games on temporal graphs}
\label{sec:games}
We focus on games with two players. Nevertheless, to 
highlight the nature of our definitions,  we 
directly introduce both games for an arbitrary number of players. Since the 
games are somewhat similar,
we start by making some general 
definitions for both temporal games before 
describing the specifics. 
For a temporal graph $\mathcal{G}=(V,(E_t)_{t=1}^\infty)$ and a number~$k\in \mathbb{N}$ 
of players, $\Diff(\mathcal{G},k)$ \big($\Vor(\mathcal{G},k)$\big) denotes the 
$k$-player \emph{temporal diffusion game} (\emph{temporal Voronoi game}) on the 
temporal graph $\mathcal{G}$, where each player has 
a distinct color in~$[k]$. Moreover, we use the color~$0$ to which 
we refer as \emph{gray}. The \emph{strategy space} of each player~$i\in [k]$ is 
the vertex set~$V$, that is,
each player~$i$ selects a single vertex $p_i\in V$, which is then immediately 
colored by 
her color~$i$.
If two players pick the same vertex, then it is colored gray.
A \textit{strategy profile} of the game is a tuple~$(p_1 , \dots, p_k ) \in 
V^k$ containing the initially chosen vertex of each player.
We also use the term \emph{position} to refer to the vertex~$p_i$ chosen by 
player~$i$. 

Now, for both games the strategy profile $(p_1 , \dots, p_k )$ determines an initial (partial) coloring of the graph at time~$t=0$.
We use $U_i^t(p_1, \dots, p_k)$ to denote the set of vertices having color~$i$ at time~$t$.
Then $U_i^0(p_1, \dots, p_k) = \{p_i\}$, unless some other player chose the same vertex, in which case $U_i^0(p_1, \dots, p_k) = \emptyset$.
In each of the two games, the coloring at time~$t$ then determines the coloring at time~$t+1$ by the specific rules described below.

Let further $u_i^t(p_1, \dots, p_k):=|U_{i}^t(p_1, \dots, p_k)|$ be the number of vertices having color~$i$ at time~$t$.
As vertices will never change their color once they have been colored,
the sequence $(U_i^t(p_1, \dots, p_k))_{t=0}^\infty$ must eventually become constant.
We denote its limit by $U_i(p_1, \dots, p_k)$.
The \textit{payoff} or \textit{outcome} for player~$i$ is then $u_i(p_1, \dots, p_k) := \abs{U_i(p_1, \dots, p_k)}$,
and the players aim to maximize this value.
Consequently, we say 
that a player $i\in [k]$ plays a \emph{best response} to the 
other 
players in the strategy profile~$(p_1, \dots, p_k)$ if for all
vertices~$p' \in V$ it holds that~$u_i(p_1, 
\dots, p_{i-1}, p', p_{i+1}, \dots, p_k) \leq u_i(p_1, \dots, p_k)$.
A strategy profile~$(p_1, \dots, p_k)$ is a \textit{Nash equilibrium} if
every player~$i \in [k]$ plays a best response to the other players.
Any strategy profile of the form~$(p, p, \dots, p)$ in which all players choose the same vertex is called~\emph{trivial}.

It remains to specify how the strategy profile $(p_1,\dots, p_k)$ 
determines the coloring of the vertices~$V$ in 
the two games.

\paragraph{Temporal diffusion games.}
In a \textit{temporal diffusion game}, the temporal graph~$\mathcal{G}$ is 
colored by the following propagation 
process over time. We call a vertex \emph{uncolored} if no color has been 
assigned to it (so far).
In step~$t$, we consider the layer~$G_t$.
We color a so far uncolored vertex~$v$ with color~$i\in [k]$ if $v$ has at 
least one 
neighbor in~$G_t$ that is colored with color~$i \in [k]$ and no neighbor in 
$G_t$ that 
is colored with 
any other color~$j \in [k] \setminus \{i\}$.
Every uncolored vertex with at least two neighbors in~$G_t$ colored by two 
different 
colors~$i,j \in [k]$ is colored gray.

\paragraph{Temporal Voronoi games.}
In a \textit{temporal Voronoi game},
at time~$t$ a vertex~$v$ is colored with color~$i\in [k]$
if~$p_i$ is the only position of a player that can reach~$v$ until time~$t$ by a temporal walk. 
(If an uncolored vertex is simultaneously reached by two or more players, then it is colored gray.)
In effect, $v \in U_i(p_1, \dots, p_k)$ if and only if
$\td(p_i,v)<\td(p_j,v)$ 
holds for all $j\neq i$.
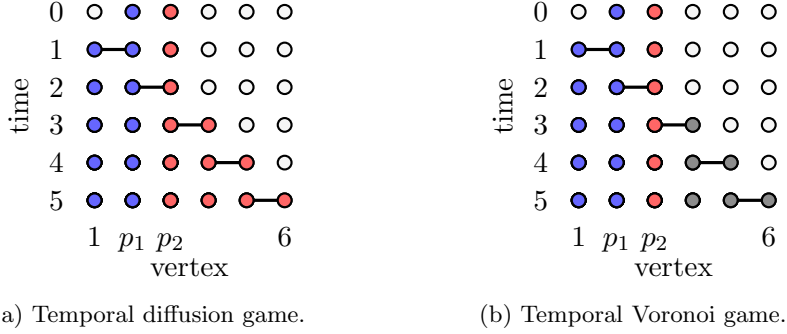
\begin{figure}[t]
	\centering
	\begin{subfigure}[t]{.4\textwidth}
		\flushleft
		\begin{minipage}[t]{\textwidth}
			\center 
			\begin{tikzpicture}[scale=0.5]
			
			\gamegrid{6}{5}
			\def\p{2}
			\def\q{3}

			\draw[edge]
				\foreach \t [evaluate={\t as \tt using int(\t+1)}] in {1,...,5} {
					(v-\t-\t) -- (v-\tt-\t)
				};
			
			\colorvertex{\p}{0}{vertexP};
			\colorvertex{\q}{0}{vertexQ};
			
			\foreach \t in {1,...,5} {
				\colorvertex{\p}{\t}{vertexP};
				\colorvertex{1}{\t}{vertexP};
				\colorvertex{\q}{\t}{vertexQ};
			}
			
			\foreach \t in {3,4,5} {
				\colorvertex{4}{\t}{vertexQ};
			}
			\foreach \t in {4,5} {
				\colorvertex{5}{\t}{vertexQ};
			}
			\colorvertex{6}{5}{vertexQ};
					
			\path[every node/.style={anchor=base}]
				(0, \vertexlabely)
				+(1, 0) node {$1$}
				+(\p, 0) node {$p_1$}
				+(\q, 0) node {$p_2$}
				+(6,0) node {$6$}
				;
			\end{tikzpicture}
			\caption{Temporal diffusion game.}
			\label{ex:diff}
		\end{minipage}
	\end{subfigure}%
	\begin{subfigure}[t]{.4\textwidth}
		\centering
		\begin{minipage}[t]{\textwidth}
			\center 
			\begin{tikzpicture}[scale=0.5]
					\gamegrid{6}{5}
			\def\p{2}
			\def\q{3}

			\draw[edge]
				\foreach \t [evaluate={\t as \tt using int(\t+1)}] in {1,...,5} {
					(v-\t-\t) -- (v-\tt-\t)
				};
			
			\colorvertex{\p}{0}{vertexP};
			\colorvertex{\q}{0}{vertexQ};
			
			\foreach \t in {1,...,5} {
				\colorvertex{\p}{\t}{vertexP};
				\colorvertex{1}{\t}{vertexP};
				\colorvertex{\q}{\t}{vertexQ};
			}
			
			\foreach \t in {3,4,5} {
				\colorvertex{4}{\t}{vertexG};
			}
			\foreach \t in {4,5} {
				\colorvertex{5}{\t}{vertexG};
			}
			\colorvertex{6}{5}{vertexG};
			
			\path[every node/.style={anchor=base}]
				(0, \vertexlabely)
				+(1, 0) node {$1$}
				+(\p, 0) node {$p_1$}
				+(\q, 0) node {$p_2$}
				+(6,0) node {$6$}
				;
			\end{tikzpicture}
			\caption{Temporal Voronoi game.}
			\label{ex:vor}
		\end{minipage}
	\end{subfigure}%
	\caption{Example of our two games on the temporal path 
	$\mathcal{G}=([6],E_1,\dots, E_5)$ with $E_t=\{\{t,t+1\}\}$ for $t\in [5]$, 
	where player~1 
	selects vertex $p_1=2$ and player~2 selects vertex $p_2=3$. In the diffusion game, the 
color of a player cannot ``pass'' a 
vertex colored by the other player. Thus, the two players split the vertices, 
that is, player~1 colors the vertices in~$[1,2]$, while player~2 colors the 
vertices in $[3,6]$. In contrast to this, in the Voronoi game, player~1 is able 
to ``catch up'' with player~2: They both reach the vertices $4$, $5$, 
and~$6$ at the same time. The displayed strategy profile is a Nash 
equilibrium in $\Vor(\mathcal{G},2)$ but not in $\Diff(\mathcal{G},2)$. }
	\label{example}
\end{figure} 

\smallskip
Note that we defined temporal diffusion games and temporal
Voronoi games such that both temporal games played on a temporal graph 
$\mathcal{G}$ with lifetime~1 are equivalent to the (non-temporal) 
game played 
on the static graph $\mathcal{G}_\downarrow$. An example of a temporal 
diffusion game and a temporal Voronoi game is shown in 
\Cref{example}. 
Notably, in the displayed temporal diffusion game each player colors a superset 
of the 
vertices they color in the temporal Voronoi game.
In fact, by definition of the two games, this holds for every temporal graph 
(as in the static case).

\paragraph{Difference games.}
We also define \emph{difference} versions of both temporal diffusion and Voronoi games, as they will prove useful in analyzing the classical versions.
We only define these for $k=2$ players, since there is no single canonical generalization to more players.
The key difference in difference games is that the payoff of player~1 becomes $\Delta(p_1, p_2) := u_1(p_1, p_2) - u_2(p_1, p_2)$ and the payoff of player~2 $\Delta(p_2, p_1) = -\Delta(p_1, p_2)$.
Note that this turns the game into a zero-sum game. 
Additionally, in these games, the second player is forbidden to select the same vertex as the first player --- otherwise she would always be able to force a tie.

We denote the difference versions of temporal diffusion and Voronoi games by $\dDi(\Gg, 2)$ and $\dVor(\Gg, 2)$, respectively.
We also set $\Delta^t(p_1, p_2) := u_1^t(p_1, p_2) - u_2^t(p_1, p_2)$.

It is interesting to note that, in two-player games where all vertices are eventually colored (possibly in gray),
difference games are equivalent to classical games in which the players are each awarded half a point for gray vertices.

\section{Temporal diffusion games} \label{se:TDG}
As already mentioned, we only consider the 2-player setting here and in the rest of the paper.
We start with temporal trees and then move to temporal cycles.

\subsection{Temporal trees} \label{se:TDG_tree}
We first prove that a Nash equilibrium is in general not 
guaranteed to exist in temporal diffusion games, not even on temporal paths. 
Afterwards, we show that as soon as a temporal tree is temporally connected, every game admits a Nash equilibrium.

We start by showing that the temporal diffusion game on the temporal path  
depicted in 
\Cref{example} (where 
every edge of the underlying graph only occurs in one layer) 
does not 
admit a Nash equilibrium. This is in contrast to the 
static case where a Nash equilibrium is guaranteed to exist on every path 
\cite[Theorem 1]{roshan}. 
\begin{theorem} \label{general_path}
	There is a temporal path~$\Pp$
	such that there is no Nash equilibrium in $\Diff(\Pp,2)$.
\end{theorem}
\begin{proof}
	We prove the theorem by showing that for the temporal path~$\mathcal{P}= 
	([6], (E_1,\ldots,E_{5}))$ with lifetime~5 and $E_t = \{ \{t, t+1\}\}$ for $t\in [5]$ (see 
\Cref{example}),
        there is no Nash equilibrium in~$\Diff(\mathcal{P},2)$.
        For the sake of contradiction, 
	assume that~$(p_1, p_2)$ with $p_1<p_2$ is a Nash equilibrium  
	in~$\Diff(\mathcal{P},2)$. Player $1$ colors all vertices 
in~$[\max(1,p_1-1),p_2-1]$ and player $2$ all vertices in $[p_2,n]$.
	
	Clearly, $p_2=p_1+1$ must hold, as otherwise player 2 can color 
	additional vertices by moving to vertex $p_1+1$. Moreover, if $p_1\geq 3$, 
        then player~1 
	can additionally color vertex~$1$ by moving to vertex~$1$. Thus, 
	$p_1\in \{1,2\}$, $p_2=p_1+1$, and thereby~$p_2\in \{2,3\}$. However, this implies 
that player~1 colors at most two 
	vertices and can increase her payoff by moving to vertex~4, thereby 
	coloring three vertices.
\end{proof}

\subsubsection{Temporally connected trees}
\label{sec:superset-trees}
If all vertices can pairwise reach each other (i.e., if the graph is temporally connected), then a Nash 
equilibrium for temporal diffusion games is guaranteed to exist on temporal trees.

An important observation to be made here is that when the two players pick vertices~$p_1$, $p_2$,
then player 1 will color (at least) all the connected components of $\footprint{\Tt} - p_1$ except for the one containing~$p_2$.
Symmetrically, this also holds for player~2.

We can thus observe the following.
\begin{lemma}\label{thm:diff-tree-neighbors}
	Let~$\Tt$ be a temporally connected tree and let $(p_1, p_2)$~be any Nash equilibrium of $\Diff(\Tt, 2)$.
	Then $p_1$~and~$p_2$ must be adjacent.
\end{lemma}
\begin{proof}
	If $p_1$, $p_2$ are not adjacent, then let $x \notin \{p_1, p_2\}$ be any vertex on the path connecting $p_1$~and~$p_2$.
	Without loss of generality, $x \notin U_1(p_1, p_2)$.
	Then $U_1(x, p_2) \supset U_1(p_1, p_2)$, i.e., player 1 can improve by moving to~$x$.
\end{proof}

Note that in the situation where both players start on adjacent vertices, the payoffs are identical to the static diffusion game played on $\footprint{\Tt}$.
From \cref{thm:diff-tree-neighbors}, it is then not hard to see that a Nash equilibrium must always exist and that it will always be located at the \emph{centroid} of $\footprint{\Tt}$
(this is a vertex~$c$ minimizing the maximum size of any component of $\footprint{\Tt} - c$).
\begin{theorem}\label{thm:diff-tree}
	Let $\Tt$ be a temporally connected temporal tree.
	Then $\Diff(\Tt, 2)$ contains a Nash equilibrium which can be computed in linear time.
\end{theorem}

The formal proof of \cref{thm:diff-tree} proceeds exactly as in the static case (\cite[Thm.~4]{roshan}), thus we omit it.
Note that the centroid of a graph can be computed in linear time by a folklore dynamic programming algorithm.

We remark that \cref{thm:diff-tree} can also be extended to disjoint unions of temporally connected trees.
We can then simply apply \cref{thm:diff-tree} to a largest tree.
If the resulting payoff for either player is less than the size of the second-largest tree, then they will simply switch to that tree.
In either case the resulting strategy profile is a Nash equilibrium.

\begin{figure}[t]
			\center 
			\begin{tikzpicture}[scale=0.7]
			\def\starttime{1}
			\gamegrid{8}{2}
			
			\draw[edge]
				\foreach \v[evaluate=\v as \w using int(\v+1)] in {3,...,7} {
					(v-\v-1) -- (v-\w-1)
					(v-\v-2) -- (v-\w-2)
				}
				(v-1-1) -- (v-2-1) -- (v-3-1)
				(v-1-2) -- (v-2-2)
				;
			
			\path[every node/.style={anchor=base}]
				(0, \vertexlabely)
				+(1, 0) node {$1$}
				+(8,0) node {$8$}
				;
			\end{tikzpicture}
			\caption{A monotonically shrinking temporal path $\mathcal{P}$ for 
which neither in $\Diff(\mathcal{P},2)$ nor in $\Vor(\mathcal{P},2)$ 
a Nash equilibrium exists.}
			\label{ex:diffms}
\end{figure}
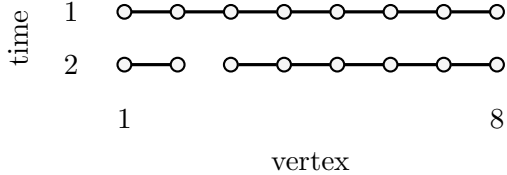
\subsubsection{Monotonically shrinking paths}
We have seen above that a Nash equilibrium is guaranteed 
to exist on temporally connected and thereby also on monotonically growing temporal paths. 
It turns out that enforcing the opposite, that is, the graph is monotonically 
shrinking, does not guarantee the existence of a Nash equilibrium:
\begin{theorem}\label{thm:shrink}
    There is a monotonically shrinking temporal path~$\Pp$ consisting of two 
layers which only differ in one edge
    such that there is no Nash equilibrium in  $\Diff(\mathcal{P},2)$.
\end{theorem}
\begin{proof}
 Let $\mathcal{P}= 
	([8], (E_1,E_2))$ (of lifetime~2) with $E_1 = \{ \{i, i+1\}\mid i\in [7]\}$ and 
$E_2=E_1\setminus \{\{2,3\}\}$ (see \Cref{ex:diffms}).
	For the sake of contradiction, assume that $(p_1,p_2)$ is a Nash 
        equilibrium in~$\Diff(\mathcal{P},2)$. We distinguish two cases:
        
1) $p_1\in[3]$ or $p_2\in[3]$: Let $p_1=i\in[3]$. Then, $p_2=i+1$, as 
this is the unique best response of player~2. Player~1 can improve by deviating 
to vertex $i+2$.

2) $p_1\in [4,8]$ and $p_2\in [4,8]$:
  If~$p_1=i\in[4,5]$,
  then $p_2=i+1$ is the unique best response of player~2 from the 
relevant interval $[4,8]$.
  Player~1 can improve by choosing vertex~3.
  If~$p_1=6$, then~$p_2=3$ is the unique best response of player~2.
  If~$p_1=i\in[7,8]$, then~$p_2=i-1$ is the unique best response of player~2 
from~$[4,8]$.
  Player~1 can improve by choosing vertex~3.
\end{proof}

Intuitively, the reason why a disappearing edge is enough to 
prevent the existence of a Nash equilibrium on a temporal path is that players 
may
want to play in the immediate surrounding of such a disappearing edge, in  
order to color some part of the temporal path that otherwise remains uncolored 
(in \Cref{ex:diffms}, these are the vertices $\{1,2\}$). However, if the 
disappearing edge is not located around the center, then the player close to 
this 
edge is at risk of 
loosing many vertices to the other player. 
Such possibly ``lost'' vertices do not appear in monotonically growing temporal graphs where edges are 
not allowed to disappear.

\subsection{Temporal cycles} \label{se:TDG_cycle}
In this section, we prove that in contrast to paths, a Nash 
equilibrium may fail to exist on temporally connected cycles and monotonically shrinking cycles. However, enforcing 
that edges do not disappear over time is enough to 
guarantee the existence of a Nash equilibrium.
\subsubsection{Temporally connected cycles}
The guaranteed existence of a Nash equilibrium on temporally connected  
trees (\cref{thm:diff-tree}) does not extend to temporally connected 
cycles despite the fact that still all vertices 
will ultimately be colored. This can be shown using the graph depicted in 
\Cref{example} with an additional layer connecting all vertices to a cycle 
and a similar argument as in \Cref{general_path}.

\begin{theorem}\label{sup_cycle}
	There is a temporally connected temporal cycle $\mathcal{C}$
	such that there is no Nash equilibrium in $\Diff(\mathcal{C},2)$.
\end{theorem}
\begin{proof}
	Let $\mathcal{C}= ([n],(E_1,\ldots,E_{n}))$ be the temporal cycle of size~$n 
	\geq 6$ with $E_t = \{ \{t, t+1\}\}$ for $t<n$ and $E_{n}=\{\{n,1\}\}\cup 
	\{\{i, i+1\}\mid i \in [n -1 ]\}$ of lifetime $n$.
	We prove the statement by showing that there is no Nash equilibrium on 
	$\mathcal{C}$.
	Let~$(p_1, p_2)$ be a strategy profile in $\Diff(\mathcal{C},2)$ with~$p_1 
	< p_2$. Then, player~$1$ colors all 
	vertices in~$[p_1,p_2 - 1]$ and player~$2$ colors all vertices in~$[p_2, 
	n]$. If~$p_1 > 1$, then player~$1$ colors vertex~$p_1 -1$. If~$p_1 >2$, 
	then the vertices~$[1,p_1 - 2]$ are distributed evenly between the players 
	with one vertex being colored gray if the number of vertices in~$[1,p_1 - 
	2]$ is odd. 
	We consider two different cases. 
	\begin{enumerate}
		\item First, assume that~$p_2 \neq p_1 +1$. By moving to vertex $p_1 
		+1$ player $2$ colors  more vertices than before. 
		\item Second, assume that~$p_2 = p_1 + 1$. If~$p_1 \geq 3$, then at 
		least one vertex in~$[1,p_1]$ is colored gray or by player~$2$.
		Thus, player $1$ can color additional vertices by moving to vertex 
		$1$. 
		Otherwise, it holds that~$p_1 \leq 2$. In this case, player~$1$ colors 
		at most two vertices. Since we assumed that~$p_2 = p_1 + 1$, it follows 
		that~$p_2 \leq 3$. By moving to vertex $4$ player~$1$ colors 
		all vertices in~$[4,n]$. Since~$n \geq 6$, these are at least three 
		vertices, so that vertex~$4$ is better for player~$1$ than her 
		current vertex.
	\end{enumerate}
\end{proof}

\subsubsection{Monotonically shrinking cycles}
The example in \Cref{ex:diffms} can be modified to show that 
enforcing edges to only disappear over time is not enough to guarantee the 
existence of a Nash equilibrium on a temporal cycle (as in 
the case of temporal paths). 
\begin{theorem}\label{thm:shrinkc}
	There is a monotonically shrinking temporal cycle~$\mathcal{C}$ with lifetime~2
	such that there is no Nash equilibrium in $\Diff(\mathcal{C},2)$.
\end{theorem}
\begin{figure}[t]
			\center 
			\begin{tikzpicture}[scale=0.7]
			\def\starttime{1}
			\gamegrid{11}{2}
			
			\draw[edge]
				\foreach \v[evaluate=\v as \w using int(\v+1)] in {4,...,8} {
					(v-\v-1) -- (v-\w-1)
					(v-\v-2) -- (v-\w-2)
				}
				(v-1-1) -- (v-2-1) -- (v-3-1) -- (v-4-1)
				(v-9-1) -- (v-10-1) -- (v-11-1)
				(v-2-2) -- (v-3-2)
				(v-1-1) -- +(-0.3, 0)
				(v-11-1) -- +(0.3, 0)
				;
			\draw[edge,dotted]
				(v-1-1) ++(-0.3,0) -- +(-0.4,0)
				(v-11-1) ++ (0.3,0) -- +(0.4,0)
				;
			
			\path[every node/.style={anchor=base}]
				(0, \vertexlabely)
				+(1, 0) node {$1$}
				+(4, 0) node {$4$}
				+(11,0) node {$11$}
				;
			\end{tikzpicture}
			\caption{A monotonically shrinking temporal cycle $\mathcal{C}$ for 
which no Nash equilibrium in $\Diff(\mathcal{C},2)$
exists.}
			\label{ex:diffmsc}
\end{figure}
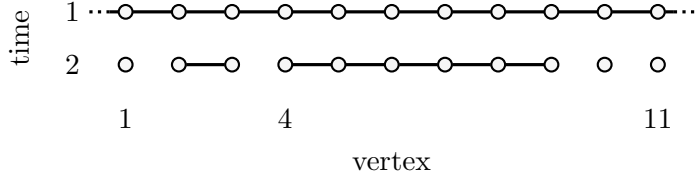

\begin{proof}
 We prove the theorem by considering a temporal cycle (see \Cref{ex:diffmsc}) similar to the 
monotonically shrinking temporal path without Nash equilibrium used in 
\Cref{thm:shrink}.

Specifically, et $\mathcal{C}=([11], (E_1,E_2))$ (of lifetime $2$) with $E_{1}=\{\{i,i+1\}\mid i\in [10] 
\}\cup \{\{11,1\}\}$ and 
$E_{2}=E_{1}\setminus\{\{1,2\},\{3,4\},\{9,10\},\{10,11\},\{11,1\}\}$. For the 
sake of contradiction, let $(p_1,p_2)$ be a Nash equilibrium in 
$\Diff(\mathcal{C},2)$. 

\begin{compactenum}
\item $p_1\in[2]$ or $p_2\in[2]$: Let $p_1=i\in[2]$. Then there exist two best 
responses of player 2, that is vertex $3$ and vertex $10$: If $p_2=3$, then 
player 1 can improve by deviating to vertex $4$. If $p_2=10$, then player $1$ 
can improve by moving to vertex $9$.
\item $p_1\in[3,4]$ or $p_2\in[3,4]$: Let $p_1=i\in[3,4]$. Then, $p_2=i+1$ is 
the unique best response of player 2. Player~1 can improve by moving to vertex 
$i+2$. 
\item $p_1=10$ or $p_2=10$: Let $p_1=10$. Then, $p_2=9$ is the unique best 
response of player 2. Player 1 can improve by moving to vertex $8$. 
\item $p_1=11$ or $p_2=11$: Let $p_1=11$. Then, $p_2=3$ or $p_2=4$ are the best 
responses of player 2. In both cases, Player 1 can improve by moving to vertex 
$5$. 
\item $p_1\in [5,9]$ and $p_2\in [5,9]$: 
  If~$p_1=i\in[5,6]$,
  then $p_2=i+1$ is the unique best response of player~2 from $[5,9]$.
  Player~1 can improve by choosing vertex~4.
  If~$p_1=7$, then~$p_2=4$ is the unique best response.
  If~$p_1\in[8,9]$, then the best responses of player~2 from~$[5,9]$ 
is~$p_2=i-1$.
  Player~1 can improve by choosing vertex~4.
\end{compactenum} 
\end{proof}
      
\subsubsection{Monotonically growing cycles} \label{subsub:moncycles}
If we require that edges do not disappear over time,
then a Nash equilibrium is again guaranteed to exist: 
\begin{theorem}\label{mono_cycle}
	On every monotonically growing 
	temporal cycle~$\mathcal{C} = ([n],(E_i)_{i\in [\tau]})$ a Nash equilibrium 
	in $\Diff(\mathcal{C},2)$ exists and
	can be found 
	in~$\mathcal{O}(\tau \cdot n)$ time.
\end{theorem}

Proving \cref{mono_cycle} will be the goal of this subsection.
We split the proof into three parts.
In Part~I, we start by first studying \diDiff games which turn out to be simpler.
In Part~II, we formally analyze the relationship between Nash equilibria in \diDiff games and Nash 
equilibria in temporal diffusion games. Generally 
speaking, we prove that for every monotonically growing 
temporal cycle $\mathcal{C}$, at least one Nash equilibrium of those found in 
$\dDi(\mathcal{C},2)$ is also a Nash equilibrium in $\Diff(\mathcal{C},2)$. 
We put together all our results in Part~III, where we describe 
an algorithm that constructs a Nash equilibrium for each monotonically growing 
temporal cycle.

The relevance of \diDiff{} games for finding Nash equilibria of temporal diffusion games can be motivated by the following observation.
\begin{observation}\label{cycle_24_l2}
For every non-trivial strategy profile $(p_1,p_2)$ and monotonically growing temporal cycle $\mathcal{C}$, no vertex is uncolored and at 
most two vertices are colored gray in $\Diff(\mathcal{C},2)$.
\end{observation}
\begin{proof}
Since~$\mathcal{C}$ is monotonically growing, layer~$\tau$ of~$\mathcal{C}$ 
is 
a cycle. Consequently, the players can eventually spread their color on the cycle
until they are stopped by a vertex that is already 
colored. Since each player starts 
from only
one position, vertices colored by the different players ``meet'' at exactly two 
places. At each place, at most 
one 
vertex can be colored gray. Clearly, no uncolored vertices remain. 
\end{proof}

In the remainder of this subsection,
we will assume that the given monotonically growing temporal 
cycle $\mathcal{C} = ([n], (E_i)_{i\in [\tau]})$ has $n > 1$~vertices
and lifetime~$\tau>1$ (otherwise insert an edgeless layer at the beginning).
Moreover, we may assume by symmetry that the edge $\{n,1\}$ only appears in layer~$\tau$, i.e., $\{n,1\}\notin E_{\tau-1}$.
We denote by
$\stopForest{} = \stopForest(\Cc)$
the temporal linear forest (i.e., disjoint union of paths) obtained from~$\Cc$ by replacing layers~$\tau, \tau+1, \dots$ by empty layers.
Recall that $\Omega_{\Gg}^{\ell}(v)$ is the set of vertices reachable from~$v$ until time~$\ell$ in $\Gg$.
Note that
\[\Omega_{\stopForest}^t(v) = \Omega_{\stopForest}^{\min\{t, \tau-1\}}(v) = \Omega_{\Cc}^{\min\{t, \tau-1\}}(v)\]
holds for all vertices~$v$ and times~$t$.
For the sake of brevity we will abbreviate $\Omega_\Cc^{\tau-1}(v) = \Omega_{\stopForest}^{\tau-1}(v)$
as $\Omega(v)$.

Recall that for two vertices $u,v\in V$ 
we say that \emph{$u$ reaches $v$ until step $\ell$} if $\td(u,v)\leq \ell$.

\paragraph{Part~I: Temporal difference diffusion games.}
 We show how to find a Nash equilibrium in a \diDiff 
game on a monotonically growing temporal cycle.

We can simplify this problem by considering \diDiff games on 
$\stopForest{}$ due to the following lemma.
\begin{lemma}\label{cycle_22_l1}
A strategy 
profile is a Nash equilibrium in $\dDi(\mathcal{C},2)$ if and only 
if it is a Nash equilibrium in $\dDi(\stopForest(\mathcal{C}),2)$.
\end{lemma}
\begin{proof}
Since each player starts spreading her color from one position and 
since layer~$\tau$ of~$\mathcal{C}$ is a cycle,
both players color the same 
number of vertices from step~$\tau$ on. 
It follows that~$\Delta^{\tau-1}(u, v) = \Delta(u,v)$ for 
all~$u,v 
\in [n]$ in $\mathcal{C}$. Consequently, in $\mathcal{C}$, a player can improve~$\Delta(p_1,p_2)$ 
by changing 
her 
position if and 
only if she can improve~$\Delta^{\tau - 1}(p_1,p_2)$ by changing her 
position.
Clearly, $\Delta^{\tau -1}$ in $\mathcal{C}$ is exactly the payoff in $\dDi(\stopForest(\Cc), 2)$.
\end{proof}
Note that when we use $\Delta^{\tau-1}(u,v)$ in the following we will not always state explicitly the graph on which we play, as this will always be $\stopForest(\Cc)$ (and we anyway have that $\Delta^{\tau-1}(u,v)$ in $\Cc$ is equal to $\Delta^{\tau-1}(u,v)$ in $\stopForest(\Cc)$).

Thus, in the following we focus on \diDiff games on $\stopForest(\Cc)$.
In preparation, we prove two elementary observations about reachability in temporal forests.
\begin{observation}\label{cycle_23_l1}
  Let~$u,v,\alpha \in [n]$ be three vertices of $\stopForest{}$ and $t \in \NN$.
  If $\at(u,\alpha)\le t$ and $\at(v,u)\le t - \abs{\alpha- u}$, then~$\at(v,\alpha)\le t$.
\end{observation}
\begin{proof}
	If $\at(u, \alpha) < \at(v, \alpha)$, then every edge between $u$ and $\alpha$ must already be present
	when $\alpha$¸ is reached from~$v$.
	Thus we then have $\at(v,\alpha) \le \at(v,u) + \abs{\alpha -u}\leq t-|\alpha-u|+|\alpha-u|=t$.
\end{proof}

\begin{observation}\label{cycle_23_l2}
	In $\stopForest{}$ it holds for all vertices~$u \leq w < v$ that
	$\at(v,w)\le \at(u,v)+ \abs{v-w}-1$.
\end{observation}\begin{proof}
	All edges between~$u$~and~$v$ (and thus between~$w$~and~$v$) are present at time~$\at(u, v)$ and all later times.
\end{proof}

It turns out that there exists a special type of 
vertices (called \emph{\nc} vertices) from which a Nash equilibrium in $\stopForest{}$ can be easily constructed.
These \nc vertices are defined as 
follows.
(Recall that a \emph{centroid} of a path is a vertex located in the middle of it; if the path has an even number of vertices then both middle vertices are centroids.)
\begin{definition}
A vertex~$v$ is called \emph{\nc} if it is a centroid of~$\Omega(v)$
and maximizes the cardinality of that set.
\end{definition}
For an example of a \nc vertex, see \Cref{cycle_22_fig10}.
We first show that at least one \nc vertex exists in $\stopForest{}$. For this, we use the following 
lemma.

\begin{lemma}\label{cycle_22_l8}
Let $v$ be any vertex and let~$m$ be a central vertex of~$\Omega(v)$.
Then, $\Omega(m) \supseteq \Omega(v)$.
\end{lemma}
\begin{proof} 
	Let~$\Omega(v) =: [\alpha, \beta]$. We can assume that~$v < m$ by symmetry.
	It suffices to show that $m$ reaches $\alpha$ and $\beta$ until time~$\tau-1$.
	In order to show this, for technical reasons, in the following, we work on the monotonically growing temporal linear 
forest $\forest{}$ obtained from~$\stopForest{}$ by replacing the empty layers~$\tau, \tau+1, \dots$ by layer~$\tau-1$ (note that $\forest$ and $\stopForest$ are identical in the first $\tau-1$ layers and thus in particular $m$ reaches the same vertices until step $\tau-1$ in both graphs). 
	By \cref{cycle_23_l2}, in $\forest{}$
	\begin{align}
	\at(m,v) \leq \at(v,m) + \abs{v-m}-1 \label{cycle_23_eq1}.
	\end{align}
	Since~$v$ reaches~$\beta$ until time~$\tau-1$ and 
	since~$v$ has to pass~$m$ in order to reach~$\beta$, we conclude that~$v$ 
	reaches~$m$ until time~$\tau-1 - \abs{m-\beta}$.
	Applying this to \Cref{cycle_23_eq1}, we get
	\begin{align}
	\at(m,v) \leq \tau -1 - \abs{m-\beta} + \abs{v-m}-1 \label{cycle_23_eq2}.
	\end{align}
	Since~$m$ is a central vertex of~$[\alpha, \beta]$, it holds 
	that~$\abs{m-\beta}+1 \geq \abs{m-\alpha}$.
	Applying this to \Cref{cycle_23_eq2}, we get
	\begin{align*}
	\at(m,v) &\leq \tau -\abs{\alpha-m} + \abs{v-m} -1 \\
	&= \tau - \abs{\alpha - v} -1. \label{cycle_23_eq3} 
	\end{align*}
	
	Since~$m$ reaches~$v$ until time~$\tau - \abs{\alpha-v}$, we 
	conclude 
	by \cref{cycle_23_l1} that~$m$ reaches $\alpha$ until time~$\tau-1$.
	Also $m$ clearly reaches $\beta$ until time~$\tau-1$ since $v$ does so and $v < m \leq \beta$.
\end{proof}

From \cref{cycle_22_l8}, the guaranteed existence 
of a \nc vertex directly follows:

\begin{lemma}\label{cycle_22_l9}
There exists at least one \nc vertex in $\stopForest{}$.
\end{lemma}
\begin{proof}
We construct a \nc vertex as follows. 
Let
\[R := \max_{v \in [n]} \; \abs{\Omega(v)}\]
be the maximum number of vertices reachable from any vertex 
in~$\stopForest{}$.
Let~$v\in [n]$ be a vertex that reaches $R$~vertices in~$\stopForest$
and let~$m$ be a central vertex of~$\Omega(v)$. By 
\Cref{cycle_22_l8}, $\Omega(m) \supseteq \Omega(v)$
and since the latter set is maximal, equality must hold.
Consequently,~$m$ is a \nc{} vertex.
\end{proof}

Intuitively, it seems to be a good strategy for a player to play on a \nc 
vertex $v$, as she colors as many vertices as possible in 
$\dDi(\stopForest{},2)$ in the absence of a second player ``stealing'' vertices from her and the other player can 
``steal'' at most half of these vertices (by playing 
adjacent to her). In fact, we will show that if a player plays on a \nc vertex, 
she is guaranteed to color at least as many vertices as the other player in 
$\dDi(\stopForest{},2)$. 

To prove this, we need the following lemma about how the sets of
vertices reachable by the players determine their payoffs.

\begin{lemma}\label{lem:payoff-new}
	Let $u,v\in [n]$ and 
	$[a_u, b_u] := \Omega(u)$ and $[a_v, b_v] := \Omega(v)$.
	If $a_u < a_v \leq b_u < b_v$,
	then in $\dDi(\stopForest{}, 2)$ we have
	$U_1(u, v) = [a_u, \ceil{x}-1]$
	and $U_2(u, v) = [\floor{x}+1, b_v]$
	where $x = (b_u + a_v)/2$.
\end{lemma}
\begin{proof}
	Since $a_v \leq b_u < b_v$, it follows that the edge $\{b_u, b_u+1\}$ appears no later than time~$\tau-1$.
	Thus, we have $\td(u, b_u) = \tau-1$ and analogously $\td(v, a_v) = \tau-1$.
	Also, we clearly have $\td(u, b_u-1) \leq \tau-2$ and $\td(v, a_v+1) \leq \tau-2$.
	If $a_v = b_u$ or $a_v + 1 = b_u$, then we are done.
	Otherwise, $a_v + 1 \leq b_u-1$ and $a_v+1$ and $b_u-1$ are the furthest vertices reached by $v$ resp.\ $u$ until time~$\tau-2$.
	In that case delete layer~$\tau-1$ from~$\stopForest{}$ and recursively apply this lemma
	to see that $U^{\tau-2}_1(u, v) = [a_u', \ceil{x'}-1]$
	and $U^{\tau-2}_2(u,v) = [\floor{x'}+1, b_v']$
	where $a_u' \leq u$, $b_v' \geq v$, and
	$x' = (b_u-1 + a_v + 1)/2 = x$.
	From this the claim follows.
\end{proof}

\begin{lemma}\label{cycle_22_l5}
Let $u<v\in [n]$
such that neither of $\Omega(u)$ and $\Omega(v)$ contains the other.
Then $\Delta^{\tau-1}(u, v)  = \abs{\Omega(u)} - \abs{\Omega(v)}$. 
\end{lemma}
\begin{proof}
	If $\Omega(u)\cap \Omega(v) = \emptyset$, then each 
	player colors exactly the vertices the player reaches and the statement follows.
	Thus, it remains to consider the case~$\Omega(u) \cap \Omega(v) \neq \emptyset$.
	Using the notation of \cref{lem:payoff-new}, this and $u<v$ implies that $a_u<a_v\leq b_u<b_v$. Thus, applying \cref{lem:payoff-new},
	we have that 
	\begin{align*}
		\Delta^{\tau-1}(u, v) &= \ceil{x - 1} - a_u + 1 - (b_v - \floor{x + 1} + 1)
		\\ &= 2x - a_u - b_v
		\\ &= b_u + a_v - a_u - b_v 
		\\ &= \abs{\Omega(u)} - \abs{\Omega(v)}.
		\qedhere
	\end{align*}
\end{proof}

Using \Cref{cycle_22_l5}, we are now ready to show that if a player 
plays on a \nc vertex, then the other player can color at most the same number of vertices.

\begin{lemma}\label{cycle_22_l6}
If~$u$ is a \nc vertex, then~$\Delta^{\tau-1}(u,v)\geq 0$ for all $v\in [n]$. 
\end{lemma}
\begin{proof}
If $u = v$ then the result is obvious.
If $u\neq v$, first consider the case that
that~$\Omega(v) \subseteq \Omega(u)$. Since~$u$ is a central vertex 
of~$\Omega(u)$, player~$1$ colors at 
least half the vertices of~$\Omega(u)$. Thereby, player~$2$ cannot color more 
vertices than player~$1$, thus~$\Delta^{\tau-1}(u,v) \geq 0$.

Otherwise, as~$u$ is a \nc vertex,
$\Omega(u)$ and $\Omega(v)$ must be incomparable.
In this case, it holds that~$\Delta^{\tau-1}(u,v) = 
|\Omega(u)| - |\Omega(v)|$ by \Cref{cycle_22_l5}. Since~$u$ is a \nc 
vertex,~$v$ cannot reach more vertices than~$u$, that is~$\Delta^{\tau-1}(u,v) \geq 0$.
\end{proof}

In the following, we construct a Nash equilibrium using \nc vertices.
With \Cref{cycle_22_l6}, it is easy to show that each strategy profile 
where both players play on different \nc vertices is a Nash equilibrium:

\begin{lemma}\label{cycle_22_l7}
 If~$u$ and~$v$ are different \nc vertices, then~$(u,v)$ is a Nash 
equilibrium in $\dDi(\stopForest{},2)$.
\end{lemma}
\begin{proof}
Since~$u$ and~$v$ are both \nc vertices, we conclude by 
\Cref{cycle_22_l6} that~$\Delta^{\tau-1}(u,v)= 0$.
Additionally, by \cref{cycle_22_l6}, 
no player can improve her payoff in $\dDi(\stopForest{},2)$ by moving to a 
different position. 
As a result,~$(u,v)$ is a Nash equilibrium.
\end{proof}

Next, we show that if only one \nc vertex exists, then the \nc vertex and 
any vertex adjacent to it form a Nash equilibrium (see \Cref{cycle_22_fig10} 
for an example). In 
order to prove that such a strategy profile is always a Nash equilibrium, we 
need the following lemma. 

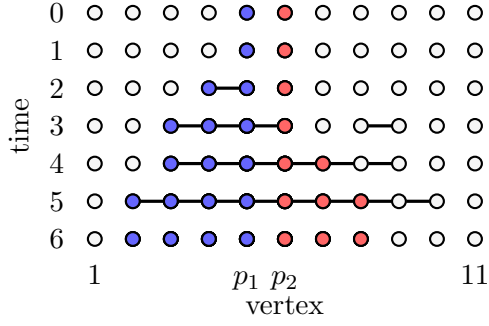
\begin{figure}[t]
\centering
\begin{tikzpicture}[scale=0.5]
	\gamegrid{11}{6}
	\def \p {5}
	\def \q {6}
	
	\draw[edge]
		(v-4-2) -- (v-5-2)
		(v-3-3) -- (v-4-3) -- (v-5-3) -- (v-6-3)
		(v-8-3) -- (v-9-3)
		\foreach \v [evaluate={\v as \w using int(\v+1)}] in {3,...,8} {
			(v-\v-4) -- (v-\w-4)
			(v-\v-5) -- (v-\w-5)
		}
		(v-2-5) -- (v-3-5)
		(v-9-5) -- (v-10-5)
		;
	
	\foreach \t in {0,...,6} {
		\colorvertex{\p}{\t}{vertexP}
		\colorvertex{\q}{\t}{vertexQ}
	}
	\foreach \t in {2,...,6} {
		\colorvertex{4}{\t}{vertexP}
	}
	\foreach \t in {3,...,6} {
		\colorvertex{3}{\t}{vertexP}
	}
	\foreach \t in {5,6} {
		\colorvertex{2}{\t}{vertexP}
		\colorvertex{8}{\t}{vertexQ}
	}
	\foreach \t in {4,5,6} {
		\colorvertex{7}{\t}{vertexQ}
	}
	
	\path[every node/.style={anchor=base}]
		(0, \vertexlabely)
		+(1, 0) node {$1$}
		+(\p, 0) node {$p_1$}
		+(\q, 0) node {$p_2$}
		+(11,0) node {$11$}
		;
	
\end{tikzpicture}
\captionof{figure}{A \diDiff game on $\stopForest{}$
where only one \nc vertex exists. The \nc vertex is 
vertex~$5$ which is the position of player~$1$. Player~$2$ colors only one 
vertex less than player~$1$ by playing directly next to her on vertex~$6$. This strategy profile is a Nash equilibrium.}
\label{cycle_22_fig10}
\end{figure} 

\begin{lemma}\label{cycle_22_l10}
Assume that $m\in [n]$ is a central vertex of $[\alpha, \beta] \subseteq \Omega(m) $
and that the number of 
vertices in~$[\alpha, \beta]$ is odd and larger than~1.
Then $\Omega(m+1) \supseteq [\alpha +1 , \beta]$.
\end{lemma}
\begin{proof}
    Clearly $\Omega(m+1) \supseteq [m+1, \beta]$.
	As~$m$ is the unique central vertex 
	of~$[\alpha, \beta]$, it holds that~$\abs{m+1 - \beta} = \abs{\alpha +1- m}$.
	Since~$m$ reaches~$\beta$ until time~$\tau-1$, it follows that~$m$ reaches~$m+1$ (and thus also $m+1$ reaches~$m$)
	no later than at time
	\[
		\tau -1 - \abs{m+1 - \beta} = \tau-1 - \abs{\alpha+1 - m}.
	\]
	From this we can derive by \cref{cycle_23_l1} that $m+1$~reaches~$\alpha+1$ until time~$\tau-1$.
\end{proof}

Using \Cref{cycle_22_l10}, we show that if
player~$1$ 
plays on the unique \nc vertex and player~$2$ plays next to 
player~$1$, then player 1 colors exactly one vertex more than player 2. Using 
this, 
we prove that the strategy profile is a Nash equilibrium.
\begin{lemma}\label{cycle_22_l11}
If there is a unique \nc vertex~$u$, then $(u, u+1)$ is a Nash equilibrium in 
$\dDi(\stopForest{},2)$. 
\end{lemma}

\begin{proof}
Let $[\alpha_u, \beta_u] := \Omega(u)$.
Then~$R:=|\Omega(u)|$ is the maximum number of 
vertices reachable from any vertex in~$\stopForest{}$.
Note that $R > 1$ since we assume $n > 1$.

Observe that $u < n$ (as $u$ is the unique nice vertex) and write $v := u + 1$.
We start by computing $\Delta^{\tau-1}(u,v)$.
Since $u$ is the only \nc vertex, by \Cref{cycle_22_l10}, $\Omega(v) \supseteq [\alpha_u + 1, \beta_u]$.
We have $\beta_u + 1 \notin \Omega(v)$ or $v$ would be a \nc{} vertex.
As $R$ is odd, player~2 colors the~$\frac{R-1}{2}$ vertices in~$[v,\beta_u]$
and player~1 colors the~$\frac{R-1}{2}+1$ vertices in~$[\alpha_u,u]$.
That is, $\Delta^{\tau-1}(u,v)= 1$. 

Suppose now for the sake of contradiction that player~$2$ can improve her 
payoff by moving from $v$ to some vertex $v' \neq u$ with $\Delta^{\tau-1}(u,v') < 1$. 
We consider two cases. 
\begin{enumerate}
\item First, assume that~$\Omega(u)$ and $\Omega(v')$ are incomparable.
By \Cref{cycle_22_l5}, it holds 
that~$\Delta^{\tau-1}(u,v') = |\Omega(u)| - |\Omega(v')|$. 
Since there is only one \nc vertex, it needs to hold 
that~$|\Omega(v')| < R$. It follows 
that~$\Delta^{\tau-1}(u,v') \geq 1$ leading to a contradiction.
\item Otherwise, it holds that~$\Omega(u) \subseteq \Omega(v')$ or~$\Omega(v') \subseteq \Omega(u)$. 
Since~$|\Omega(u)| = R$, it cannot hold that~$\Omega(u) \subset \Omega(v')$. 
Consequently,~$\Omega(v') \subseteq \Omega(u)$.  Since $u$ is the central vertex of~$\Omega(u)$,
player~1 colors at least~$\frac{R-1}{2}+1$ vertices and player~2 colors at most~$\frac{R-1}{2}$
vertices, that is,~$\Delta^{\tau-1}(u,v') \geq 1$, which is a contradiction.
\end{enumerate}

For the sake of contradiction, assume that player~$1$ can improve her 
payoff by moving from $u$ to $u'$, that is,~$\Delta^{\tau-1}(u',v) >
1$. 
We consider two cases. 
\begin{enumerate}
\item First, assume again that $\Omega(u')$ and $\Omega(v)$ are incomparable.
By \Cref{cycle_22_l5}, it holds 
that~$\Delta^{\tau-1}(u',v) = |\Omega(u')| - |\Omega(v)|$.
We showed before that~$|\Omega(v)| \ge R -1$. Since~$R$ 
is the maximum number of vertices reachable from any vertex in~$\stopForest$, 
it holds that $|\Omega(u')|  \leq R$.
Hence,~$\Delta^{\tau-1}(u',v) = |\Omega(u')| - |\Omega(v)| \leq 1$
leading to a contradiction.

\item Otherwise, it holds that~$\Omega(u') \subseteq \Omega(v)$ or~$\Omega(v) \subseteq \Omega(u')$.
First, assume that~$\Omega(u') \subseteq \Omega(v)$.
Recall that we have argued above that $[\alpha_u+1,\beta_u]\subseteq \Omega(v)$.
Since~$v=u+1$, $v$ reaches $\frac{R-1}{2}-1$ vertices to the right of $v$ 
and at least~$\frac{R-1}{2}$ vertices to the left of~$v$ in $\stopForest{}$. Thus, 
independent of whether~$u'$ is left or right of $v$, player~2 colors at 
least $\frac{R-1}{2}$ vertices and, as $\Omega(u') \subseteq \Omega(v)$, player~1
colors at most   $\frac{R-1}{2}+1$ vertices. It follows  that~$\Delta^{\tau-1}(u',v) \leq 1$. 
Otherwise,~$\Omega(v) \subset \Omega(u')$, which implies that~$\Omega(v)=[\alpha_{u}+1,\beta_u]$ and $\Omega(u')=\Omega(u)$ since~$u$ is the unique \nc vertex.
Again, player~2 colors at least $\frac{R-1}{2}$ vertices and player~1 colors at most~$\frac{R-1}{2}+1$ vertices.
\qedhere
\end{enumerate}
\end{proof}

We now have everything at hand to find a Nash 
equilibrium in $\dDi(\stopForest{},2)$.
By \cref{cycle_22_l1}, this result can be directly
extended to 
\diDiff games on monotonically growing temporal cycles.

\begin{lemma}\label{cycle_22_l14}
If~$u$ and~$v$ are different \nc vertices, 
then~$(u,v)$ is a 
Nash equilibrium in $\dDi(\mathcal{C},2)$.

Otherwise, let~$v$ be the unique \nc vertex. 
Then,~$(v, v+1)$ is a Nash equilibrium in $\dDi(\mathcal{C},2)$.
\end{lemma}
\begin{proof}
By \cref{cycle_22_l9,cycle_22_l7,cycle_22_l11},
the claims hold for~$\dDi(\stopForest{},2)$.
By \cref{cycle_22_l1}, this result transfers to~$\dDi(\Cc, 2)$.
\end{proof}

\paragraph{Part~II: From difference diffusion games to diffusion 
games.}
In Part~I, we described how to find Nash equilibria in \diDiff 
games on monotonically growing temporal cycles based on the notion of \nc 
vertices. We now prove that, for each monotonically growing temporal cycle 
$\mathcal{C}$, at least one Nash equilibrium in 
$\dDi(\mathcal{C},2)$ described in \Cref{cycle_22_l14} is also a Nash 
equilibrium in $\Diff(\mathcal{C},2)$.

To this end, we prove that in several cases a Nash equilibrium $(u,v)$ 
in $\dDi(\mathcal{C},2)$ colors at most one vertex in gray. This directly 
implies that $(u,v)$ is also a Nash equilibrium in $\Diff(\mathcal{C},2)$ 
as proven in the following lemma: 
\begin{lemma}\label{cycle_24_l1}
If $(u,v)$ is a Nash equilibrium in 
$\dDi(\mathcal{C},2)$
which results in at most one gray vertex,
then $(u,v)$ is a Nash equilibrium in 
$\Diff(\mathcal{C},2)$.
\end{lemma}
\begin{proof}
Without loss of generality and for the sake of contradiction, assume that 
player~$1$ gets a higher payoff by moving to vertex~$u'$. Since~$(u, 
v)$ is a Nash equilibrium in $\dDi(\mathcal{C},2)$, the payoff of player~$2$ 
increases by at 
least the same amount. However, since all but at most one vertex is colored by 
one of 
the 
players for~$(u,v)$ in $\Diff(\mathcal{C},2)$, it is not possible that both 
players color an additional vertex.
\end{proof}

We observe that if the two players play next to each other, then 
at most one vertex is colored gray.
This yields two easy cases based on \Cref{cycle_22_l14} in which a Nash 
equilibrium exists: 
\begin{lemma} \label{le:27}
\begin{itemize}
 \item[1.] If $v$ is the unique \nc vertex, then 
$(v,v+1)$ is 
a Nash equilibrium in $\Diff(\mathcal{C},2)$. 
\item[2.] Let $\Omega$ be a maximum-size set of vertices reachable from a 
vertex in $\stopForest{}(\mathcal{C})$. If $|\Omega|$ is even,
then its two central vertices form a Nash equilibrium in 
$\Diff(\mathcal{C},2)$. 
\end{itemize}
\end{lemma}
\begin{proof}
  In both cases the players play on adjacent vertices and the resulting 
  strategy profile is a Nash equilibrium in $\dDi(\mathcal{C},2)$ by 
  \cref{cycle_22_l14}.
  As the players start on adjacent vertices, there exists at most one gray vertex.
  Thus \cref{cycle_24_l1} implies the claim. 
\end{proof}

It remains to consider the case where there exist two \nc vertices in 
$\stopForest(\mathcal{C})$
and the size of the maximum set of vertices reachable by any vertex in 
$\stopForest(\mathcal{C})$ is odd. In this 
case, there do not 
always exist two adjacent \nc vertices.
However, we can prove that in this case there either exists a pair of \nc vertices such 
that at 
most one vertex is colored gray if the players select them or all pairs of \nc 
vertices are Nash equilibria. To be able to make this 
distinction, 
we define the distance between two vertices as follows:
\begin{definition}
Let~$u,v \in [n]$ with~$u<v$. The distances between~$u$ and~$v$ are 
defined 
as~$d_1(u,v) := v-u$ and~$d_2(u,v) := n - (v - u)$.
\end{definition}

We can show that we get two gray vertices on a monotonically 
growing 
temporal cycle $\mathcal{C}$ where the maximum number of vertices reachable 
from any vertex in~$\stopForest{}(\mathcal{C})$ is odd if and only if the 
distances between the 
positions of the players are even in both directions:

\begin{lemma}\label{cycle_24_l4}
Let $u, v$ be two distinct \nc vertices,
each reaching an odd number of vertices in~$\stopForest$.
Then, the strategy profile $(u,v)$ in $\Diff(\mathcal{C},2)$
results in two gray vertices if and only if
both~$d_1(u,v)$ and~$d_2(u,v)$ are even.
\end{lemma}
\begin{proof}
Say $u < v$.
We will show that one vertex of $[u, v]$ is colored gray if and only if $d_1(u, v)$ is odd.
The same arguments can then be also applied to the complement part of $\Cc$.

By assumption, $\abs{\Omega(u)} = \abs{\Omega(v)}$ is odd and $u, v$ are the central vertices of $\Omega(u)$, $\Omega(v)$ respectively.
In particular, the sets
$\Omega'_u := \Omega(u) \cap [u, v]$
and
$\Omega'_v := \Omega(v) \cap [u, v]$
have the same size.

Now if $\Omega'_u \cap \Omega'_v \neq \emptyset$, then the claim follows from \cref{lem:payoff-new} (with the possibly gray vertex being~$x$).
Otherwise, until time~$\tau -1$ the two players color exactly $\Omega'_u$ resp.\ $\Omega'_v$.
At each subsequent time step each player reaches exactly one new vertex from the set~$[u, v] \setminus (\Omega(u) \cup \Omega(v))$.
If $d_1(u, v)$ is odd, then this set has odd size, implying that eventually its central vertex will be colored gray.
Otherwise, if $d_1(u,v)$ is even, then this set has even size so no vertex will be colored gray.
\end{proof}

From \Cref{cycle_24_l4} we can derive the following:
\begin{lemma}\label{le:31}
Let $u, v$ be two distinct nice vertices, each reaching an odd number of vertices in~$\stopForest{}$.
If $d_1(u,v)$ or $d_2(u,v)$ are odd, 
then $(u,v)$ is a Nash 
equilibrium in $\Diff(\mathcal{C},2)$.
\end{lemma}
\begin{proof}
	By \Cref{cycle_24_l4}, we can conclude that the strategy profile $(u,v)$ 
	results in only at most one vertex being colored in gray in 
	$\Diff(\mathcal{C},2)$. Moreover, by \Cref{cycle_22_l14}, $(u,v)$ is a Nash 
	equilibrium in $\dDi(\mathcal{C},2)$. Applying \Cref{cycle_24_l1}, it 
	follows that $(u,v)$ is a Nash equilibrium in $\Diff(\mathcal{C},2)$.
\end{proof}

By \cref{le:27,le:31} it only remains to consider the case 
where there exist two or more \nc{} vertices, each reaching an odd number of vertices in~$\stopForest{}$
and where all pairs of \nc vertices have even distance in both directions. We 
will prove that in this case, in fact, every pair of \nc vertices is a Nash 
equilibrium: 

\begin{lemma}\label{cycle_24_l5}
Let the \nc vertices be $v_1, \dots, v_z$ where $z \ge 2$.
Assume that each of them reaches an odd number of vertices in $\stopForest{}$
and that for all~$i\neq j \in [z]$, 
the distances~$d_1(v_i,v_j)$ and~$d_2(v_i,v_j)$ are even. 
Then,~$(v_i,v_j)$ is a Nash equilibrium in $\Diff(\mathcal{C},2)$ for all $i\neq j \in [z]$.
\end{lemma}
\begin{proof} 
Fix some pair of \nc vertices $u, v$. 
For the sake of contradiction, assume that~$(u,v)$ is not a Nash 
equilibrium in $\Diff(\mathcal{C},2)$. Without loss of generality, assume that 
player~$1$ can increase her payoff by $x$ when moving to vertex~$u'$. 

First, observe that by \Cref{cycle_24_l4} and \Cref{cycle_24_l2}, the 
strategy profile $(u,v)$ results in exactly two vertices being colored in gray 
and 
no vertices being uncolored in  $\Diff(\mathcal{C},2)$.

Since~$u$ and~$v$ are both \nc vertices and both 
players color the same number of vertices in $\Diff(\mathcal{C},2)$ after step 
$\tau-1$, we 
conclude by 
\Cref{cycle_22_l14} that~$(u,v)$ is a Nash equilibrium in 
$\dDi(\mathcal{C},2)$ and by \Cref{cycle_22_l6} that $\Delta(u,v)=0$ in 
$\Diff(\mathcal{C},2)$.
As $(u,v)$ is a Nash equilibrium in $\dDi(\mathcal{C},2)$, no player can deviate to color more vertices than the other player. Thus, for player $1$ to be able to increase her payoff by $x$ by moving to $u'$, also the payoff for 
player~$2$ for $(u',v)$ has to increase by~$x$.
This directly implies that $x=1$, as only two vertices 
are colored gray for the strategy profile~$(u,v)$. Moreover, this implies that 
$\Delta(u',v)=\Delta(u,v)=0$ in $\Diff(\mathcal{C},2)$  
and that
the strategy profile~$(u',v)$ results in
all vertices being colored with color $1$ or $2$ in $\Diff(\mathcal{C},2)$.

We distinguish two cases both leading to a 
contradiction to the observations from above.
\begin{enumerate}
\item First, assume that~$\Omega(u') \subseteq \Omega(v)$ or~$\Omega(v) \subseteq 
\Omega(u')$. Since $v$~is \nc{}, this can only mean that~$\Omega(u') \subseteq 
\Omega(v)$. As player~$2$ plays on the only central vertex of~$\Omega(v)$ and 
since both 
players color the same number of vertices in $\Diff(\mathcal{C},2)$ after step 
$\tau-1$, it follows that player~$2$ colors more vertices than 
player~$1$ in $\Diff(\mathcal{C},2)$. This contradicts that~$\Delta(u',v) = 0$  
in $\Diff(\mathcal{C},2)$.
\item Otherwise, $\Omega(u')$ and $\Omega(v)$ are incomparable. 
By \Cref{cycle_22_l5} and since both 
players color the same number of vertices in $\Diff(\mathcal{C},2)$ after step 
$\tau-1$, it holds 
that~$0 = \Delta(u',v) = |\Omega(u')| - 
|\Omega(v)|$ in $\Diff(\mathcal{C},2)$.
Let~$m$ be the central vertex of~$\Omega(u')$. By 
\Cref{cycle_22_l8}, $\Omega(m) = \Omega(u')$.
Thus, by \Cref{cycle_22_l5}, it holds that 
$\Delta^{\tau-1}(m,v) = |\Omega(u')| - |\Omega(v)|=0$.
As the edge~$\{1,n\}$ appears only in layer~$\tau$, at most one vertex can be colored 
gray in the first $\tau-1$ steps. As at most one vertex can be colored gray after 
$\tau-1$ steps, as $\Delta^{\tau-1}(u',v) = 
\Delta^{\tau-1}(m,v) =0$, and as for both $(u',v)$ and $(m,v)$ exactly the 
vertices 
in $\Omega(u') \cup \Omega(v)$ are colored after $\tau-1$ steps, it follows 
that the coloring of the vertices after $\tau-1$ steps for 
$(m,v)$ and for $(u',v)$ in $\Diff(\mathcal{C},2)$ needs to be the same. This 
also directly implies that the coloring of all vertices for $(m,v)$ 
and for $(u',v)$ at the end of $\Diff(\mathcal{C},2)$  needs to be the same. 
Since the strategy profile~$(u',v)$ results in no vertices being colored gray, 
the same holds for the strategy profile~$(m,v)$. However,~$m$ and~$v$ are 
both \nc vertices. By our initial assumption, both $d_1(m,v)$ 
and $d_2(m,v)$ are even. Thus, by \Cref{cycle_24_l4}, it follows that the 
strategy profile $(m,v)$ results in two vertices being colored gray, which 
contradicts our previous observation.
\end{enumerate}
Thus,~$(u, v)$ is a Nash equilibrium.
\end{proof}

Since we have exhausted all cases, we showed that there is a Nash equilibrium in 
every temporal diffusion game on a monotonically growing temporal cycle.
In the next part, we summarize our result in an algorithm.

\paragraph{Part~III: The algorithm.}
\begin{algorithm}[t] \label{al:1}
	\caption{Compute a Nash equilibrium on a monotonically growing 
	cycle.}\label{algo}
	\begin{minipage}{\textwidth}
	\textbf{Input}: Monotonically growing temporal 
	cycle~$\mathcal{C} =([n],(E_i)_{i\in [\tau]})$.
		\\
	\textbf{Output}: A Nash equilibrium in 
$\Diff(\mathcal{C},2)$.
	\\[1em]
	$R:= $ maximum number of vertices that are reachable from any 
	vertex in~$\stopForest{}(\mathcal{C})$. 
	\begin{enumerate}[\hspace{0pt}(1)]
			\item\label{algo1} If~$R$ is even, then return the two central 
vertices of a size-$R$ vertex set that is reachable from some vertex 
in~$\stopForest{}(\mathcal{C})$. 
			\item\label{algo2} If~$R$ is odd, then let~$v_1, v_2, \dots, v_z$ be the 
			\nc{} vertices.
			\begin{enumerate}
				\item\label{algo2b} If~$z \geq 2$:
				\begin{enumerate}
					\item\label{algo2bi} If there is a pair~$v_i\neq v_j$ such 
					that~$d_1(v_i, v_j)$ or~$d_2(v_i, v_j)$ is odd, then 
					return~$(v_i,v_j)$. 
					\item\label{algo2bii} Otherwise, return $(v_1,v_2)$. 
				\end{enumerate}
				\item\label{algo2a} Otherwise, return $(v_1,v_1+1)$.
			\end{enumerate}
		\end{enumerate}
	\end{minipage}
\end{algorithm}

We claim that \Cref{algo} returns a Nash equilibrium for every temporal diffusion game on a 
monotonically growing temporal cycle $\mathcal{C} =([n],(E_i)_{i\in [\tau]})$ in $\mathcal{O}(n\cdot \tau)$ time.
We show the correctness of \Cref{algo} in \Cref{cycle_25_l1} and prove its 
running time in \Cref{cycle_25_l2}.

\begin{lemma}\label{cycle_25_l1}
\Cref{algo} returns a Nash equilibrium for every temporal diffusion game on a 
monotonically growing temporal 
cycle~$\mathcal{C}$.
\end{lemma}
\begin{proof}
Note first that by \Cref{cycle_22_l14}, every strategy profile returned 
by \Cref{algo} is a Nash equilibrium in $\dDi(\mathcal{C},2)$. In the 
following, we iterate 
over all cases of \Cref{algo} and argue that \Cref{algo} always returns a Nash 
equilibrium in $\Diff(\mathcal{C},2)$.

In Case~\ref{algo1} and 
Case~\ref{algo2a}, \Cref{le:27} applies.  
Otherwise, a strategy profile returned by \Cref{algo} falls under 
Case~\ref{algo2b}. Consequently, there are at least two \nc vertices.
In Case~\ref{algo2bi}, the players play on 
\nc vertices with odd distance in at least one direction. Thus, 
\Cref{le:31} 
applies.  Considering Case~\ref{algo2bii}, all pairs of
\nc vertices have even distance in both 
directions. Thus, \Cref{cycle_24_l5} applies. 
\end{proof}

\begin{lemma}\label{cycle_25_l2}
Let~$\mathcal{C}=([n],(E_i)_{i\in [\tau]})$ be a monotonically growing temporal 
cycle of 
size~$n$. 
\Cref{algo} on~$\mathcal{C}$ runs in~$\bigO(n \cdot \tau)$ time.
\end{lemma}

\begin{proof}
First, we need to compute~$R$. To do so, we compute for each vertex 
$v$ the set $\Omega(v)$ of vertices reached in $\stopForest{}(\mathcal{C})$. This 
can 
be done in $\mathcal{O}(\tau)$ time for each vertex.
Subsequently, we compute the central vertices of each distinct set of vertices 
of size $R$ reachable by some vertex, which can each be done in constant time. 

Additional computation is only needed in Case~\ref{algo2b}, where we know that 
$R$ is odd and~$z \geq 2$. If the number~$n$ of vertices in~$\mathcal{C}$ is 
odd, then, for each $u,v\in [n]$,
either~$d_1(u,v)$ or~$d_2(u,v)$ is odd. Consequently, if~$n$ is odd, 
then~$(v_1,v_2)$ is a solution for Case~\ref{algo2bi}. Otherwise,~$n$ is even. 
It follows that for all~$v_i,v_j$ with~$i \neq j \in [z]$,~$d_1(v_i,v_j)$ is 
even if and only if~$d_2(v_i,v_j)$ is even. Thus, it is enough to consider the 
distances between two \nc vertices in one direction. Without loss of 
generality, assume that~$v_1 < v_i$ with~$i \in [2,z]$. We compute all 
distances~$d_1(v_1,v_i)$. If~$d_1(v_1,v_i)$ is odd, then we found a solution. 
Otherwise,~$d_1(v_1,v_i)$ is even for all~$i \in [2,z]$. 
However, this directly implies that~$d_1(v_j,v_i)$ is even for all~$i\neq j \in 
[2,z]$. Thus, all 
\nc vertices have even distance in both directions, so that 
Case~\ref{algo2bii} holds. We summarize that for Case~\ref{algo2bi} and 
Case~\ref{algo2bii} we only have to compute the distances between~$v_1$ and all 
other \nc vertices~$v_i$ with~$i \in [2,z]$. This can be done 
in~$\mathcal{O}(n)$ time. 

Altogether, \Cref{algo} runs in~$\mathcal{O}(\tau \cdot n)$ time.
\end{proof}

\noindent\Cref{mono_cycle} now directly follows from \Cref{cycle_25_l1} and 
\Cref{cycle_25_l2}.

\section{Temporal Voronoi games}\label{se:vor}

In this section, we study temporal Voronoi games.
In contrast to temporal diffusion games, here, the color of a vertex $v$ is 
determined solely by the temporal distances from the players' positions to $v$.

\subsection{Monotonically shrinking paths and cycles}\label{se:vor1}

In \Cref{se:prel}, we observed that temporal diffusion games and temporal 
Voronoi games might already differ on a simple temporal path. In contrast to 
this, both games are equivalent on 
monotonically shrinking forests  and cycles,
as no foremost walk ever needs to wait at any vertex in these graphs.

\begin{lemma}\label{thm:shrinking_equivalence}
Let $\mathcal{G}= (V,(E_i)_{i\in [\tau]})$ be a monotonically shrinking forest or cycle and let
$p_1,p_2\in V$. For the strategy profile $(p_1,p_2)$, the final coloring of the 
vertices is the same in $\Diff(\mathcal{G},2)$ and in $\Vor(\mathcal{G},2)$.
\end{lemma}

\begin{proof}
	As already noted in \Cref{se:prel}, a vertex $v$ colored with color 
$i\in[2]$ in $\Vor(\Gg,2)$ is colored the same in $\Diff(\Gg,2)$,
as $p_i$ reaching $v$ first implies that $p_i$ also reaches every vertex on a 
foremost walk from $p_i$ to $v$ first.

To see that on monotonically shrinking  forest and cycles the 
converse also holds,
assume that $v$ gets colored with color $1$ in $\Diff(\Gg,2)$.
Note that there is exactly one temporal walk from $p_1$ to $v$ which does not 
use vertex~$p_2$ and no vertex repeatedly.
Since no foremost walk ever needs to wait in $\mathcal{G}$ (as no new edges appear over time), this temporal walk
must have fewer 
edges than any temporal walk from $p_2$ to $v$.
For the same reason, the walk from $p_1$ to $v$ consists of exactly $\td(p_1, 
v)$ edges. Thus,~$\td(p_1, 
v)<\td(p_2, v)$.
\end{proof}

In particular, using \cref{thm:shrinking_equivalence}, we can transfer \cref{thm:shrink} and \cref{thm:shrinkc} to temporal 
Voronoi games:

\begin{corollary}\label{thm:shrinkVor}
    There is a monotonically shrinking temporal path~$\Pp$ and a monotonically 
shrinking temporal cycle $\Cc$ both consisting of two 
layers
    such that there is no Nash equilibrium in $\Vor(\Pp,2)$ and no Nash 
equilibrium in $\Vor(\Cc,2)$.
\end{corollary}

\subsection{Temporally connected paths and cycles}\label{se:vor2}
In contrast 
to temporal diffusion games, a 
Nash equilibrium in a temporal Voronoi game on a temporally connected 
path may fail to exist. In fact, 
the underlying dynamics of temporal Voronoi games on temporally connected paths 
might be quite intriguing and far more complex than for temporal diffusion 
games (as highlighted in the next subsection).

\begin{theorem}\label{voronoi_nonex}
There is a temporally connected path~$\Pp$ and a temporally connected cycle~$\Cc$ 
such 
that there is no Nash equilibrium in $\Vor(\Pp,2)$ and no Nash 
equilibrium in $\Vor(\Cc,2)$.
\end{theorem}
\begin{proof}
Let $\Pp=([8], (E_1,E_2))$ be the temporal path from \Cref{ex:diffms}.
By \cref{thm:shrinking_equivalence} and the proof of \cref{thm:shrink}, there 
is no Nash equilibrium in  
$\Vor(\Pp, 2)$.
Note that for all sufficiently large~$N$,
modifying layer~$N$ of~$\Pp$ does not affect the dynamics of the Voronoi game:
At that point in time,
either every vertex has already been reached by some player,
or both players have reached the same set of vertices,
therefore any vertices left uncolored can only become gray.
In particular, we may replace layer~$N$ and all subsequent layers by a complete path or cycle,
which proves the claim.
\end{proof}

\subsection{Monotonically growing trees}
\label{sec:voronoi_growing}\label{se:vor3}
We now turn to monotonically growing temporal trees and prove that every temporal Voronoi game on such a tree admits a Nash equilibrium.
In the next \Cref{sec:voronoi-cycles}, we reuse some ideas presented here to also prove a result about temporal Voronoi games on monotonically growing cycles. 
Specifically, this subsection is devoted to proving the following theorem:

\begin{restatable}{theorem}{voronoithm}\label{thm:voronoi-tree}
On every monotonically growing temporal tree $\Tt$ a Nash equilibrium in 
$\Vor(\Tt, 2)$ exists and can be found in $\bigO(n^2)$ time.
\end{restatable}

In order to prove \cref{thm:voronoi-tree}, we first analyze the best responses of a player to a given strategy of the other player. For this, we introduce the concept of boundaries.

We call an edge~$B = \{b, b'\} \in E(\Gg)$ a \emph{boundary} of a vertex~$v$
if a player positioned at~$v$ can cross~$B$ at the earliest possible time,
i.e., if ${\min\{\at(v, b), \at(v, b')\} < \appear{B}}$.
To better distinguish them from vertices, we will use capital letters for boundaries.

Let $\Bb$ be the set of all boundaries of some vertex~$v$.
Then the connected components of $\footprint{\Gg} - \Bb$ are called the \emph{boundary components} of~$v$.
For any vertex~$w$, we define~$\J_v(w)$ to be the boundary component of~$v$ that contains~$w$.
Note that $\J_v(v) = \{v\}$.
The following observation is the basis for our study of Nash equilibria.
Note that it does not only hold for monotonically growing trees, but all temporal graphs.
Compare also~\cref{fig:boundary}.

\begin{lemma}\label{thm:boundary-components}
Let $v$, $w$ be any two vertices of a temporal graph~$\Gg$.
Then $U_1(w, v) \subseteq \J_v(w)$.
\end{lemma}
\begin{proof}
To reach any vertex~$x \notin \J_v(w)$, the player positioned at~$w$ must cross some boundary of~$v$.
Since the player at~$v$ can cross that boundary at the same time or earlier,
$w$~cannot reach any vertices beyond that boundary before~$v$ does.
\end{proof}

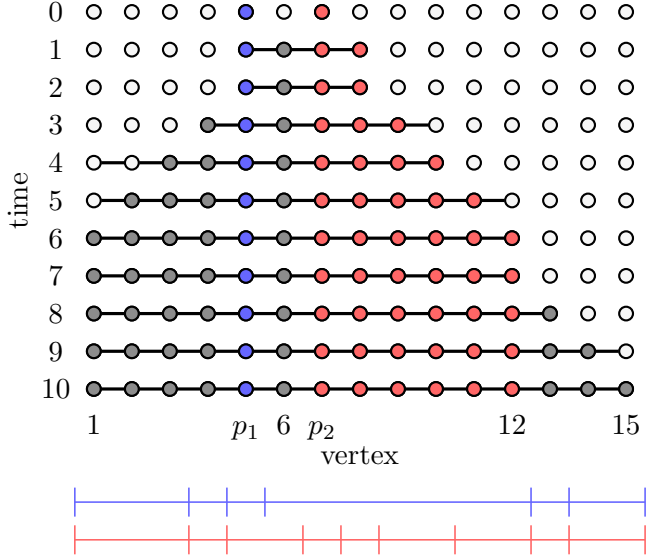
\begin{figure}[t]
	\centering
	\begin{tikzpicture}[scale=0.5]
		\gamegrid{15}{10}
		
		\def \time {10}
					
		\def \p {5}
		\def \q {7}
		
		\draw[edge]
			\foreach \t in {1,...,\time} {
				(v-\p-\t) -- (v-6-\t) -- (v-\q-\t) -- (v-8-\t)
			}
			\foreach \t in {3,...,\time} {
				(v-4-\t) -- (v-\p-\t)
				(v-8-\t) -- (v-9-\t) -- (v-10-\t)
			}
			\foreach \t in {4,...,\time} {
				(v-1-\t) -- (v-2-\t) -- (v-3-\t) -- (v-4-\t)
			}
			\foreach \t in {5,...,\time} {
				(v-10-\t) -- (v-11-\t) -- (v-12-\t)
			}
			(v-12-8) -- (v-13-8)
			\foreach \t in {9,10} {
				(v-12-\t) -- (v-13-\t) -- (v-14-\t) -- (v-15-\t)
			}
		;
		
		\colorvertex{\p}{0}{vertexP}
		\colorvertex{\q}{0}{vertexQ}
		\colorvertex{6}{1}{vertexG}
		
		\colorvertex{8}{1}{vertexQ}
		\colorvertex{9}{3}{vertexQ}
		\colorvertex{10}{4}{vertexQ}	
		\colorvertex{11}{5}{vertexQ}
		\colorvertex{12}{6}{vertexQ}
		
		\colorvertex{13}{8}{vertexG}
		\colorvertex{14}{9}{vertexG}
		\colorvertex{15}{10}{vertexG}
		
		\colorvertex{4}{3}{vertexG}
		\colorvertex{3}{4}{vertexG}
		\colorvertex{2}{5}{vertexG}
		\colorvertex{1}{6}{vertexG}
			
		\path[every node/.style={anchor=base}]
			(0, \vertexlabely)
			+(1, 0) node {$1$}
			+(\p, 0) node {$p_1$}
			+(6, 0) node {$6$}
			+(\q, 0) node {$p_2$}
			+(12, 0) node {$12$}
			+(15,0) node {$15$}
		;
			
		\draw[timeline,color=\colorP, every node/.style={vboundary,rotate=90}]
			(0.5, -3) node {}
			-- (3.5, -3) node {}
			-- (4.5, -3) node {}
			-- (5.5, -3) node {}
			-- (12.5, -3) node {}
			-- (13.5, -3) node {}
			-- (15.5, -3) node {}
		;
		\draw[timeline,color=\colorQ, every node/.style={vboundary,rotate=90}]
			(0.5, -4) node {}
			-- (3.5, -4) node {}
			-- (4.5, -4) node {}
			-- (6.5, -4) node {}
			-- (7.5, -4) node {}
			-- (8.5, -4) node {}
			-- (10.5, -4) node {}
			-- (12.5, -4) node {}
			-- (13.5, -4) node {}
			-- (15.5, -4) node {}
		;
	\end{tikzpicture}
	\caption{Example Voronoi game on a monotonically growing temporal path graph.
		On the bottom, the boundary components of $p_1=5$ (upper row, blue)
		and $p_2=7$ (lower row, red) are indicated.
		As guaranteed by \cref{thm:boundary-components},
		the sets of vertices colored by the players satisfy
		$\{p_1\} = U_1(p_1, p_2) \subseteq \J_{p_2}(p_1) = \{p_1, 6\}$
		and 
		$\ints{p_2}{12} = U_1(p_2, p_1) \subseteq \J_{p_1}(p_2) = \ints{6}{12}$.
	}
	\label{fig:boundary}
\end{figure}

Subsequently, we always assume that the graph $\Tt$ is a 
monotonically growing tree~$\Tt = (V, (E_i)_{i\in 
[\tau]})$.
For two vertices $v, w$ we will use the notation~$[v, w]$ to refer to the unique (static) path from~$v$ to~$w$ in~$\footprint{\Tt}$.
We say that a vertex~$x$ or an edge~$e$ is \emph{between}~$v$~and~$w$ if it is part of $[v, w]$.

Complementary to~\cref{thm:boundary-components}, we now find that boundaries are the \emph{only} way for a player to ``catch up'' with another player.
An important special case of the following lemma is the situation where $x = x'$.
(For an example, consider \cref{fig:boundary} with $v = p_1$ and $x = x' = p_2$.)
\begin{lemma}\label{thm:boundary}
Let $v, w$ be two vertices of~$\Tt$
and $x \in [v, w]$.
Assume~$x'$ to be any vertex with $\at(x', x) < \at(v, x)$.
Then $\at(v, w) \geq \at(x', w)$,
and equality holds if and only if
$[x, w]$~contains a boundary of~$v$.
\end{lemma}
\begin{proof}
It is clear that $\at(v, w) \geq \at(x', w)$ since any temporal path from~$v$ to~$w$ must pass~$x$,
and $x'$ reaches~$x$ before~$v$~does.

Furthermore, if the path~$[x, w]$ contains a boundary~$B$ of~$v$,
then let $\at(v, B)$ be the time $v$~reaches either endpoint of~$B$.
Since $B$~is a boundary of~$v$, $\at(v, B) < \appear{B}$.
Then $\at(x', x) < \at(v, x) \leq \at(v, B) < \appear{B}$,
therefore $v$, $x$, and $x'$ must all be on the same side of~$B$ (with $w$ on the other side).
Thus a temporal path from~$x'$ to~$w$ must cross~$B$ before reaching~$w$.
Since $B$~is a boundary of~$v$, this implies that $x'$~cannot reach~$w$ earlier than~$v$~does.

Conversely, assume that $\at(v, w) = \at(x', w)$.
Then let $w'$ be the first vertex on~$[x,w]$ with $\at(v, w') = \at(x', w')$.
Since $\at(v, x) > \at(x', x)$, we have $w' \neq x$.
So there is a vertex~$y$ of~$[x, w]$ right before~$w'$.
Since $\at(x', y) \leq \at(v, y) -1 \leq \allowbreak \at(v, w')-2 = \at(x', w')-2$,
a foremost path from~$x'$ to~$w'$ must ``wait'' at least one time step before crossing~$\{y, w'\}$.
This proves that that path crosses~$\{y, w'\}$ at the earliest possible time,~$\appear{\{y, w'\}}$.
This makes $\{y, w'\}$ a boundary of~$x'$ and thus also of~$v$,
since $\at(v, w') = \at(x', w')$ implies that $v$~and~$x'$ cross $\{y, w'\}$ simultaneously.
\end{proof}

With $\footprint{\Tt}$ being a tree, every boundary splits it in two connected components,
about which we can make the following observation.

\begin{lemma}\label{thm:boundary-half}
Let $v$ be a vertex of~$\Tt$ and~$B$ a boundary of~$v$.
Let $C$~and~$C'$ be the two connected components of $\footprint{\Tt}-\{B\}$,
with $v \in C$.
Then $\at(v, w) \leq \at(w', w)$ for any vertices $w \in C$ and $w' \in C'$.
\end{lemma}
\begin{proof}
Let $B = \{b, b'\}$ with $b \in C$ and $b' \in C'$.
Then $b \in [w', w]$ and $B \in [w', b]$.
Since $B$ is a boundary of~$v$, we have $\at(v, b) < \appear{B} \leq \at(w', b)$.
Thus, $\at(v, w) \leq \at(w', w)$ by \cref{thm:boundary}.
\end{proof}

For any boundary component~$C$ of a vertex~$v$,
we call the vertex~$x \in C$ that minimizes $\at(v, x)$ the \emph{entry vertex} of~$C$
and the last edge of $[v, x]$ the \emph{entry boundary} of~$C$ (note that this edge is a boundary of~$v$). 

\begin{lemma}\label{thm:boundary-component-entry}
Let $v$ be any vertex, $C$~be any boundary component of~$v$ with entry vertex~$x$.
Then for any vertex~$w$,
$U_1(w, v) = C$ if and only if $x \in U_1(w, v)$.
\end{lemma}
\begin{proof}
The forward direction is immediate since $x \in C$.
For the reverse, note that by \cref{thm:boundary-components}
we must have $U_1(w, v) \subseteq \J_v(w) = C$.
So it only remains to prove that $U_1(w, v) \supseteq C$.
So let $y \in C$ be any vertex.
By definition of~$C$, $[x, y]$~does not contain a boundary of~$v$.
Thus we may apply~\cref{thm:boundary} to get $\at(v, y) > \at(w, y)$.
\end{proof}

From \cref{thm:boundary-component-entry} we can derive the following about best responses.
\begin{lemma}\label{thm:best-response-boundaries}
Let $x$~be a vertex of~$\Tt$ and $y$ a best response to~$x$.
Then $U_1(y, x) = \J_x(y)$ and all edges leaving~$\J_x(y)$ are boundaries of~$y$.
\end{lemma}
\begin{proof}
By \cref{thm:boundary-components} it holds that $U_1(y, x) \subseteq \J_x(y)$.
Let $b$~be the entry vertex of~$\J_x(y)$,
then $U_1(b, x) = \J_x(y)$ by~\cref{thm:boundary-component-entry}.
Thus $U_1(y, x) = \J_x(y)$, or $b$~would have been a better response to~$x$ than~$y$.
Since $\at(y, v) < \at(x, v)$ for all~$v \in \J_x(y)$, all the edges leaving~$C$ (which must be boundaries of~$x$ by definition) are boundaries of~$y$, too.
\end{proof}
Looking back at~\cref{fig:boundary}, we can thus deduce that $p_2$~was not a best response to~$p_1$. Instead, the red player should have played at~$6$.

Now we can already show the following powerful structural result.
\begin{lemma}\label{thm:boundary-component-structure}
Let $x$ be any vertex of~$\Tt$ and $y$~a best response to~$x$.
Then $U_1(y, x) = \J_x(y)$ and $U_1(x, y) = \J_y(x)$ are disjoint sets, connected by an edge.
Furthermore, $V \setminus (\J_x(y) \cup \J_y(x))$ can be partitioned into sets
which are boundary components of both, $x$~and~$y$.
\end{lemma}
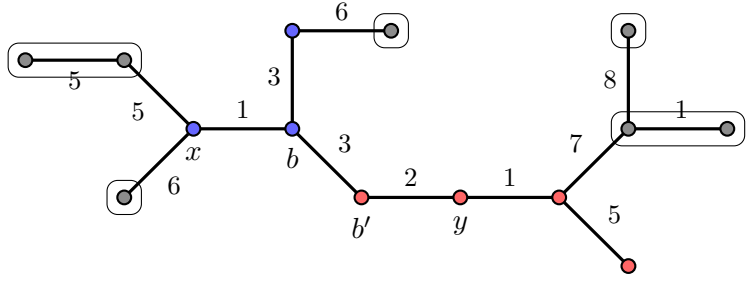
\begin{figure}
	\centering
	\begin{tikzpicture}[scale=1.3]
		\path[edge,auto]
			(0,0) node[vertexP,label=below:$x$] (x) {}
			-- node[timelabel] {1}
			++(1,0) node[vertexP,label=below:$b$] (b) {}
			-- node[timelabel] {3}
			++(0.7,-0.7) node[vertexQ,label=below:$b'$] (b') {}
			-- node[timelabel] {2}
			++(1,0) node[vertexQ,label=below:$y$] (y) {}
			-- node[timelabel] {1}
			++(1,0) node[vertexQ] (yy) {}
			-- node[timelabel] {7}
			++(0.7, 0.7) node[vertexG] (z) {}
			-- node[timelabel] {8}
			++(0, 1) node[vertexG] (z1) {}
			(z)
			-- node[timelabel] {1}
			++(1, 0) node[vertexG] (z2) {}
			(yy)
			-- node[timelabel] {5}
			++(0.7, -0.7) node[vertexQ] {}
			(b)
			-- node[timelabel] {3}
			++(0, 1) node[vertexP] {}
			-- node[timelabel]{6}
			++(1, 0) node[vertexG] (x3) {}
			(x)
			-- node[timelabel] {6}
			++(-0.7, -0.7) node[vertexG] (x2) {}
			(x)
			-- node[timelabel]{5}
			++(-0.7, 0.7) node[vertexG] (x1) {}
			-- node[timelabel]{5} ++(-1, 0) node[vertexG] (x1b) {}
			;
		\node[fit=(x1)(x1b),group] {};
		\node[fit=(x2),group] {};
		\node[fit=(x3),group] {};
		\node[fit=(z)(z2),group] {};
		\node[fit=(z1),group] {};
	\end{tikzpicture}
	\caption{
		Example illustrating~\cref{thm:boundary-component-structure};
		each edge is labeled with its appearance time.
		Here, $y$ is a best response to~$x$,
		$U_1(x, y)$~is colored~blue,
		and $U_1(y, x)$ is colored~red.
		The boxes mark common boundary components of~$x$~and~$y$.
	}
	\label{fig:boundary-component-structure}
\end{figure}
\begin{proof}
Let $B = \{b, b'\}$ be the entry boundary of~$\J_x(y)$ with $b' \in \J_x(y)$ being the entry vertex (see also~\cref{fig:boundary-component-structure}).
By \cref{thm:best-response-boundaries}, $U_1(y, x) = \J_x(y)$ and  $B$ is a boundary of~$y$, too.
Note that $[x, b] \subseteq U_1(x, y)$.
In particular, $y$~cannot have any boundaries in~$[x, b]$.
Thus, $b$ is the entry vertex of~$\J_y(x)$ and $U_1(x, y) = \J_y(b) = \J_y(x)$ by \cref{thm:boundary-component-entry}.

It remains to show that $\J_x(v) = \J_y(v)$ for any $v \in V \setminus (\J_x(y) \cup \J_y(x))$.
Say without loss of generality that $v$ is on the same side of~$B$ as~$x$
and let~$B'$~be the entry boundary of~$\J_y(v)$.
By the above, $x$~and~$y$ are on the same side of~$B'$,
and by \cref{thm:boundary-half} $B'$ is also a boundary of~$x$ (due to $B$ being a boundary of~$x$).
Thus, $\at(x, w) = \at(y, w)$ for all $w \in \J_y(v)$.
Therefore $\J_y(v)$ must also be a boundary component of~$x$.
\end{proof}

Note that the fact that $\J_x(y)$~and~$\J_y(x)$ are connected by an edge in \cref{thm:boundary-component-structure}
implies that for any $v \in V \setminus (\J_x(y) \cup \J_y(x))$ the two (identical) boundary components $\J_x(v)$ and~$\J_y(v)$ also have the same entry boundary.

Since $(v, w)$ forms a Nash equilibrium if and only if $v$~and~$w$ are mutual best responses,
our strategy of proving \cref{thm:voronoi-tree} will be  to consider a best response dynamic where both players alternatingly pick a best response to the other's previous choice.
We will then show that this process must eventually reach a Nash equilibrium.
For this, we require the following lemma.

\begin{lemma}\label{thm:voronoi-tree-part1}
Let $x$ be any vertex of~$\Tt$, $y$ a best response to~$x$,
and $z$ a best response to~$y$.
Let~$B$ be the entry boundary of~$\J_x(y)$.
If $x$~and~$z$ are on the same side of~$B$, then $(x, y)$ or~$(y,z)$ is a Nash equilibrium.
\end{lemma}
\begin{proof}
By \cref{thm:boundary-component-structure}, $B$ is a shared boundary of~$x$~and~$y$.
Also $U_1(x, y) = \J_y(x)$, so if $z \in \J_y(x)$ then $U_1(z, y) \subseteq \J_y(z) = \J_y(x) = U_1(x, y)$,
which would make~$x$ a best response to~$y$, proving the claim.

So assume now $z \notin \J_y(x)$, i.e., $\J_y(z) \neq \J_y(x)$.
By \cref{thm:boundary-component-structure},
it holds that $\J_x(z) = \J_y(z)$ and both of these have the same entry boundary, $B'$~say
(see \cref{fig:voronoi-tree-part1}).
Let $C, C'$ be the two connected components of~$\footprint{\Tt} - \{B'\}$, with $z \in C'$ and $\{x, y\} \subseteq C$.
Furthermore, by \cref{thm:boundary-half}, every boundary of~$z$ in~$C$ is also a boundary of~$x$ and~$y$.
Also, $B$~is clearly no boundary of~$z$ as $B$ is between~$y$ and~$B'$.
Thus, $\J_z(y) \supseteq \J_x(y) \cup \J_y(x)$.

\begin{figure}
	\centering
	\begin{tikzpicture}[scale=1]
		\path[edge]
			(0, 0) node[vertex, label=left:$x$] (x) {}
			++(1, 0.5) node[vertex] (bx) {}
			-- node[above] {$B$} 
			++(1, 0) node[vertex] (by) {}
			++(1.2, -0.3) node[vertex,label=left:$y$] (y) {}
			;
		\path[path]
			(x)
			++(-0.8, 0.8) node[vertex] (bbx) {}
			-- ++(-1.5, 0.5) node[vertex] (bbx2) {}
			;
		\path[edge]
			(bbx2)
			-- node[above]{$B'$}
			++(-1.2, 0) node[vertex] (bbz) {}
			++(-0.5, -1) node[vertex, label=left:$z$] (z) {}
			;
		\node[blob,fit={(x)(bx)(bbx)},text depth=1.3cm] (Jyx) {$\J_y(x)$};
		\node[above=2mm of Jyx] (text) {$\J_z(x) = \J_z(y)$};
		\node[blob,fit={(y)(by)},text depth=1.2cm] (Jxy) {$\J_x(y)$};
		\node[blob,fit={($(z)+(-2,-0.5)$)(bbz)($(bbz)+(0,0.5)$)}] {$\J_x(z) = \J_y(z)$};
		\node[blob,fit={(Jxy)(Jyx)(bbx2)(text)}] {};
	\end{tikzpicture}
	\caption{Illustration of the proof of \cref{thm:voronoi-tree-part1}.}
	\label{fig:voronoi-tree-part1}
\end{figure}
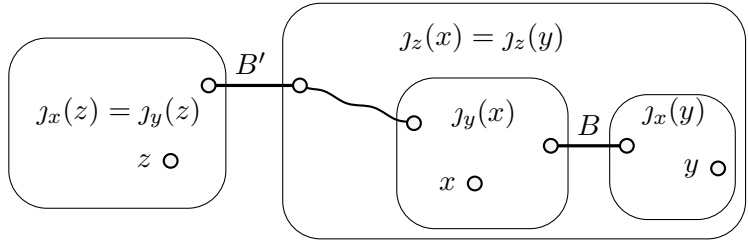

By \cref{thm:best-response-boundaries}, $B'$ is a boundary of~$z$.
Thus, by \cref{thm:boundary-half}, all boundaries of~$x$ in~$C'$ are also boundaries of~$z$.
So if some~$y' \in C'$ is a best response to~$z$, then $\J_z(y') \subseteq \J_x(y')$.
But then we must have $\abs{\J_z(y')} \leq \abs{\J_x(y')} \leq \abs{\J_x(y)} < \allowbreak \abs{\J_x(y) \cup \J_y(x)} \leq \abs{\J_z(y)}$ by choice of~$y$.
This contradicts~$y'$~being a best response to~$z$.

So any best response~$y'$ to~$z$ must necessarily lie in~$C$.
We may assume $y' \notin \J_z(y) \supseteq \J_x(y) \cup \J_y(x)$, or $y$ would already be a best response to~$z$, making~$(y, z)$~a Nash equilibrium.
Let $B''$ be the entry boundary of~$\J_z(y')$.
Since $\J_z(y)$~touches~$B'$, both, $x$~and~$z$, are on the same side of~$B''$.
Also, $B''$~is a boundary of~$x$ by \cref{thm:boundary-half}.
Therefore $\J_z(y') = \J_x(y')$.
Consequently, $\abs{\J_z(y')} = \abs{\J_x(y')} \leq \abs{\J_x(y)} < \abs{\J_x(y) \cup \J_y(x)} \leq \abs{\J_z(y)}$ by choice of~$y$,
thus contradicting the choice of~$y'$.
\end{proof}

We can now prove the existence of a Nash equilibrium in $\Vor(\Tt)$.

\begin{lemma}\label{thm:voronoi-tree-converges}
Let~$x_0$~be any vertex of~$\Tt$.
Define $x_{i+1}$ iteratively to be the entry vertex of a largest boundary component of~$x_i$.
Then there is some index~$j < \abs{V}$ such that $(x_j, x_{j+1})$ forms a Nash equilibrium.
\end{lemma}
\begin{proof}
By \cref{thm:boundary-components} and \cref{thm:boundary-component-entry}, each $x_{i+1}$~is a best response to~$x_i$.

For any~$i$, let~$B_i$~be the entry boundary of $\J_{x_i}(x_{i+1})$ and $C_i$, $C_i'$ the two connected components of~$\footprint{\Tt} - \{B_i\}$,
with $x_i \in C_i$.

By \cref{thm:voronoi-tree-part1}, if $x_{i+2} \in C_i$, then we have found a Nash equilibrium.
Thus, we may assume that~$x_{i+2} \in C_i'$ for all~$i < \abs{V}$.
Therefore, also $B_{i+1} \subseteq C_i'$ and $C_{i+1}' \subset C_i'$.
Since the sequence $C_0' \supset C_1' \supset C_2' \supset \dots$ must terminate after at most~$\abs{V}$ steps,
this concludes the proof.
\end{proof}

From \cref{thm:voronoi-tree-converges}, we can easily deduce our main result~\cref{thm:voronoi-tree}.
\begin{proof}[Proof of \cref{thm:voronoi-tree}]
Compute a sequence $x_0, x_1, \dots$ as in \cref{thm:voronoi-tree-converges} until we encounter a Nash equilibrium.
Observe that computing the reach times from any given vertex~$v_i$ to all other vertices takes $\bigO(\abs{V})$~time,
since we assume that the first appearance of any edge can be found in constant time.
Consequently, we can also determine all boundaries of~$v_i$ and thus compute~$v_{i+1}$ in $\bigO(\abs{V})$~time.
By \cref{thm:voronoi-tree-converges}, we will find a Nash equilibrium after at most $\abs{V}$~iterations,
and we can easily test for this by comparing $\abs{\J_{x_i}(x_{i-1})}$ and $\abs{\J_{x_i}(x_{i+1})}$.
Thus, we need $\bigO(\abs{V}^2)$~time overall.
\end{proof}
    
We close this section by remarking that the existence of Nash equilibria extends to monotonically growing forests as well.
To see this, simply apply the argument of \cref{thm:voronoi-tree-converges} to a connected component of maximum size.
If at any point the best response to~$x_i$ is not~$x_{i+1}$ but some vertex $y$ located in a different connected component,
then $(x_i, y)$~already forms a Nash equilibrium.

\subsection{Monotonically growing cycles}
\label{sec:voronoi-cycles}
In the previous subsection we have seen that Nash equilibria always exist in temporal Voronoi games on monotonically growing trees. 
Now we extend this result also to monotonically growing cycles.
However, the proof becomes more involved here. 

In this subsection, let $\Cc = (\zeroto{n-1}, (E_i)_{i \in \oneto{\tau}}))$ be a monotonically growing temporal cycle graph (note that it will be convenient here to number the vertices starting with 0).
We split the proof into two parts. 
In the first part, in \cref{se:vor4}, we make some structural observations that allow us to classify all Nash equilibria in temporal difference Voronoi games.
Subsequently, these structural results also help us to find Nash equilibria in (normal) temporal Voronoi games in~\cref{sec:voronoi-cycle-normal}.

\subsubsection{Temporal difference Voronoi games}\label{se:vor4}
As was the case for temporal diffusion games,
we will first look into \dVoronoi{} games,
which will help us to subsequently find Nash equilibria for Voronoi games.

For vertices, we use chains of inequalities like $a < b < c$
to express that when starting from vertex~$a$ and moving in positive direction on~$\Cc$ (that is, in increasing order of vertices), then vertex~$b$ is encountered before~$c$.
We denote subintervals of the cycle by $[a, b] \coloneqq \{v \in V \mid a \leq v \leq b\}$, thus $[a, b] \cup [b, a] = V$.
For any $x \in \ZZ$, we use $\mymod{x}$ to denote the unique $\mymod{x} \in V = \zeroto{n-1}$ with $\mymod{x} \equiv x \pmod{n}$.
Note that $\abs{[a, b]} = \mymod{b-a}+1$.

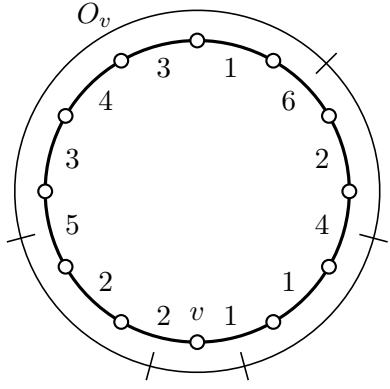
\begin{figure}
\centering
\begin{tikzpicture}[scale=2]
\path[edge,radius=1]
	(270:1) node[label=90:$v$] {}
	\foreach[count=\i] \label in {1, 1, 4, 2, 6, 1, 3, 4, 3, 5, 2, 2}{
		arc[start angle=270+(\i-1)*30, end angle=270+\i*30]
		node[vertex, fill=white] {}
		node[at={(270+\i*30-15 : 0.85)}] {$\label$}
	};
\newcommand{\rad}{1.2}
\path[timeline,radius=\rad]	
	(285:\rad) node[vboundary,rotate=285] {}
	arc[start angle=285, end angle=345]
	node[vboundary,rotate=345] {}
	arc[start angle=345, end angle=405]
	node[vboundary,rotate=405] {}
	arc[start angle=405, end angle=555]
	node[at=(480:\rad+0.15)] {$O_v$}
	node[vboundary,rotate=555] {}
	arc[start angle=555, end angle=615]
	node[vboundary,rotate=615] {}
	arc[start angle=615, end angle=645]	
	;
\end{tikzpicture}
\caption{Example for a monotonically growing cycle, with edge labels indicating appearance times.
The boundary components of~$v$ are marked on the outside, with~$O_v$ labeled.
Apart from the trivial boundary component~$\{v\}$, $v$~has three inner boundary components, all of size~$2$.
If we take the positive direction to be counterclockwise,
then~$v$ has one left and two right boundary components.
}
\label{fig:vor-ne-cycle:example}
\end{figure}

The definition of a boundary carries over from \cref{sec:voronoi_growing}.
However, we will now distinguish between left and and right boundaries.
An edge $B = \{b, \mymod{b+1}\}$ is a \emph{right boundary} of a vertex~$v$ if $\at(v, b) < \appear{B}$
and a \emph{left boundary} if $\at(v, \mymod{b+1}) < \appear{B}$.
Note that there might be one boundary that is both a left and a right boundary,
i.e., where $v$ reaches both $b$ and $\mymod{b+1}$ before $\appear{\{x,x+1\}}$.

Clearly, every boundary component~$C$ of~$v$ is now adjacent to exactly two boundaries of~$v$.
If both of these are left (resp.\ right) boundaries, then we call~$C$ a left (resp.\ right) boundary component.
Otherwise, if $C$~is not the trivial boundary component~$\{v\}$, then we call $C$ the (unique) \emph{outer boundary component} of~$v$, denoted $O_v$.
If $v$ has no such boundary component, then we set $O_v = \emptyset$.

All boundary components except~$O_v$ are called \emph{inner boundary components}.
See \cref{fig:vor-ne-cycle:example} for an example.
Note that $w \in O_v$ is intuitively equivalent to the fact that $v$ will never ``catch up'' with $w$ from behind (see also \cref{fig:cycle-catch-up}).
Due to this, we observe the following.
(Recall that~$\Omega^t(v)$ is the set of all vertices reachable from~$v$ until time~$t$.)

\begin{lemma}\label{thm:outer-property}
If $v$ is a vertex of~$\Cc$ and $w \in O_v$,
then for all~$t$
\[
	\Omega^t(v) \supseteq \Omega^t(w) \iff \Omega^t(v) = V.
\]
\end{lemma}
\begin{proof}
Assume $\Omega^t(v) \supseteq \Omega^t(w)$.
By monotonicity, it suffices to consider the minimal such~$t$,
thus we may assume that there is $x \in \Omega^{t-1}(w) \setminus \Omega^{t-1}(v)$.
Since $x \in \Omega^t(v)$, there must be a neighbor~$x'$ of~$x$ with $x' \in \Omega^{t-1}(v)$ and $\appear{\{x,x'\}} \leq t$.
Say without loss of generality $x' = \mymod{x-1}$.
If $\mymod{x+1} \in \Omega^t(w) \subseteq \Omega^t(v)$,
then~$\Omega^t(v) = V$, since $v$~cannot have reached~$\mymod{x+1}$ via~$x$ yet.
Otherwise ($\mymod{x+1} \notin \Omega^t(w)$), as $x\in \Omega^t(w)$,  we must have~$\appear{\{x,\, \mymod{x+1}\}} > t$, making this edge a right boundary of~$v$ with $v<w<\{x,\mymod{x+1}\}$,
which would contradict $w \in O_v$.

The reverse implication is trivial.
\end{proof}

\begin{figure}
\centering
\begin{tikzpicture}[scale=1.7]
\path[edge,radius=1]
	(270:1)
	\foreach[count=\i] \label in {1, 1, 3, 2, 3, 3, 7, 7}{
		arc[start angle=270+(\i-1)*45, end angle=270+\i*45]
		node[vertex, fill=white] {}
		node[at={(270+\i*45-22.5 : 0.85)}] {$\label$}
	}
	(270:1) node[vertex,fill=\colorP,label=90:$v$] {}
	(135:1) node[vertex,fill=\colorQ,label=320:$w$] {}
	;
\newcommand{\rad}{1.2}
\path[timeline,radius=\rad,color=\colorP]	
	(-90:\rad) node[vboundary,rotate=-90] {}
	arc[start angle=-90, end angle=180]
	node[vboundary,rotate=180] {}
	;
\renewcommand{\rad}{1.4}
\path[timeline,radius=\rad,color=\colorQ]
	(-45:\rad) node[vboundary,rotate=-45] {}
	arc[start angle=-45, end angle=180]
	node[vboundary,rotate=180] {}
	;

\end{tikzpicture}
\caption{Example instance with $\Omega^6(v)$ (blue, inner arc) and $\Omega^6(w)$ (red, outer arc) indicated.
Note how $\Omega^6(v) \supseteq \Omega^6(w)$ since $v$ has ``caught up'' with~$w$ in positive (counterclockwise) direction.
By \cref{thm:outer-property} we can thus conclude that $w \notin O_v$.
(In fact, $O_v = \{\mymod{v-1}\}$.)
}
\label{fig:cycle-catch-up}
\end{figure}
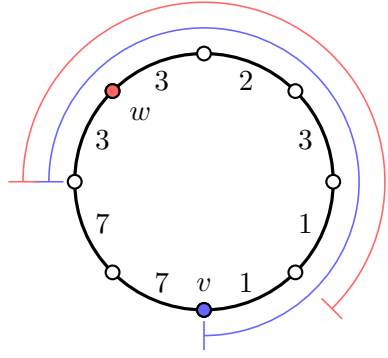

Boundaries keep their importance also on cycles, since \cref{thm:boundary-components} also applies to cycles.
We also observe the following:
\begin{lemma}\label{thm:boundary_inheritance-cycle}
Let $u < v \leq w$ be vertices of $\Cc$.
If $\{w,\, \mymod{w+1}\}$ is a right boundary of $u$, then it is also a right boundary of $v$.
\end{lemma}
\begin{proof}
 Because $\{w,\, \mymod{w+1}\}$ is a right boundary of $u$, $\at(u,w)<\appear{\{w,\, \mymod{w+1}\}}$. 
 Thus, as $u < v \leq w$, we have $\at(v,w)\leq \at(u,w)<\appear{\{w,\, \mymod{w+1}\}}$, implying that $w$ is also a right boundary of $v$.
\end{proof}
Analogously, it holds for $w\leq u<v$ that if $\{\mymod{w-1},\,w\}$ is a left boundary of~$v$, then it is also a left boundary of~$u$.

Lastly, we will also make use of the following fact.
\begin{lemma}\label{thm:boundary_inheritance2-cycle}
Let $u, v$ be vertices of $\Cc$.
If $B$ is a common right boundary of $u$ and $v$ with $u < v < B$, then every edge $C$ with $B < C < u$ 
is a right boundary of $u$
if and only if it is a right boundary of $v$.
\end{lemma}
\begin{proof}
 Let $B=\{b,\,\mymod{b+1}\}$ and let $C=\{c,\,\mymod{c+1}\}$ be a right boundary of $u$ or $v$. 
 This implies that $B$ is not a left boundary of $u$ or $v$ and thus $\at(u,\mymod{b+1})=\at(v,\mymod{b+1})=\appear{\{b,\,\mymod{b+1}\}}$ and further that $\at(u,c)=\at(v,c)$,
 which establishes the claim.
\end{proof}

As indicated by the following lemma, playing in the outer component of your opponent is often a good idea:
\begin{lemma}\label{thm:outer-beats-inner}
Let $v$, $w$ be vertices of~$\Cc$.
If $v \notin O_w$ and $w \in O_v$, then there is some time~$t$
with $\Omega^t(v) \subset \Omega^t(w)$.
\end{lemma}
\begin{proof}
  Let $\J_w(v) = [B, B']$ be a left boundary component of~$w$ without loss of generality.
  Note that~$B$ is also a left boundary of~$v$ by \cref{thm:boundary_inheritance-cycle}.
Set $B \eqqcolon \{\mymod{b-1}, b\}$ and $t \coloneqq \at(w, b)$.
Then $\Omega^t(v) \subseteq \Omega^t(w)$.

Assume for contradiction that $\Omega^t(v) = \Omega^t(w)$.
Thus, by \cref{thm:outer-property}, $\Omega^t(w) = \Omega^t(v) = V$.
In particular $\at(v, \mymod{b-1}) \leq t$, therefore $B$ is also a right boundary of~$v$.
Thus, $O_v = [B,B] = \emptyset$, but $w \in O_v$.
\end{proof}

If both players play in their opponent's outer components, then each of them is guaranteed to win all their respective inner components.
\begin{lemma}\label{thm:mutual-outer-inner}
Let $v$, $w$ be vertices of~$\Cc$.
If $v \in O_w$ and $w \in O_v$,
then $U_1(v, w) \supseteq V \setminus O_v$.
\end{lemma}
\begin{proof}
Recall that $V \setminus O_v$ is the union of all inner boundary components of~$v$.
Let $I = [B, B']$ be any inner boundary component of~$v$, say without loss of generality a right one.
Then $w \in O_v \subseteq [B', v]$.
Since $B'$ is a boundary of~$v$ and $v \in O_w$, we have $\at(v, u) < \at(w, u)$ for all $u \in [v, B'] \supseteq I$.
\end{proof}

We call an edge~$\{x,\, \mymod{x+1}\}$ a left (resp.\ right) \emph{blocker} for the vertex~$v$ if $\at(v,\, x+1) < \appear{\{x,\, x+1\}}-1$
(resp.\ $\at(v, x) < \appear{\{x,\, x+1\}}-1)$, i.e., if a temporal path from~$v$ to~$x$ needs to wait at least one time step at~$x+1$.
Clearly, blockers are a special case of boundaries.

To keep track of the situation even after a player has reached all vertices of~$\Cc$,
using the following step count functions is sometimes preferable to just using~$\Omega^t$.
For any $t\geq 0$, $v \in \oneto{n}$,
the number of steps that one can move from~$v$ to the left until time~$t$ is
\[
\lreach^t(v) \coloneqq \begin{cases}
	0 & \text{if } t \leq 0,\\
	\lreach^{t-1}(v) & \text{if } 0 < t < \appear{\{\mymod{v-\lreach^{t-1}(v)-1},\: \mymod{v-\lreach^{t-1}(v)}\}},\\
	\lreach^{t-1}(v)+1 & \text{otherwise}.
\end{cases}
\]
Symmetrically, the number of possible steps to the right is
\[
\rreach^t(v) \coloneqq \begin{cases}
	0 & \text{if } t \leq 0,\\
	\rreach^{t-1}(v) & \text{if } 0 < t < \appear{\{\mymod{v+\rreach^{t-1}(v)},\: \mymod{v+\rreach^{t-1}(v)+1}\}},\\
	\rreach^{t-1}(v)+1 & \text{otherwise}.
\end{cases}
\]
Note that we do not stop once $v$ has reached $v$ again, i.e., $\lreach^t(v)$ can be larger than $n$.
Note further that the second case in the above definitions only occurs if the respective edge is a blocker of~$v$.
Also, setting $\reach^t(v) \coloneqq 1 + \lreach^t(v) + \rreach^t(v)$, the following is true for all $t \geq 0$ (both cases may hold simultaneously):
\begin{equation}
\reach^t(v) = \begin{cases}
	\abs{\Omega^t(v)}  & \text{if } \Omega^t(v) \neq V \text{ or there is an edge~$e$ with } \appear{e} > t, \\
	\reach^{t-1}(v) + 2 & \text{if there is no edge~$e$ with } \appear{e} > t.
\end{cases}
\label{eq:reach}
\end{equation}

Using these notions, we can give the following easy characterization of boundaries.
\begin{lemma}\label{thm:boundary-reach-lemma}
Two distinct vertices $v$, $w$ of~$\Cc$ have a common right boundary~$B=\{b, \mymod{b+1}\}$ with $v < w < B$ if and only if
\[
	\rreach^t(v) = \rreach^t(w) + \mymod{w-v}
\]
holds for some time~$t$.
In that case, equality holds for all $t \geq \at(v, b)$.

Analogously, $v$~and~$w$ have a common left boundary~$B = \{\mymod{b-1},b\}$ with $B < w < v$ if and only if
\[
	\lreach^t(v) = \lreach^t(w) + \mymod{v-w}
\]
holds for some time~$t$. In that case, equality holds for all $t \geq \at(v, b)$.

\end{lemma}
\begin{proof}
We only prove the first half, as the second half follows by symmetry.

If $B$~is a right boundary of~$v$ and~$w$ with $v < w < B$,
then the player starting at~$w$ will be forced to wait at~$b$ until time~$\appear{B}-1 \geq \at(v, b)$.
Thus, eventually
\[
	\mymod{v + \rreach^{\at(v, b)}(v)} = b = \mymod{w + \rreach^{\at(v, b)}(w)}.
\]
It is further clear that $\mymod{v + \rreach^t(v)} = \mymod{w + \rreach^t(w)}$ must then also hold for all $t > \at(v, b)$.

For the converse, assume now~$t$ to be chosen minimally such that
\[
	\rreach^{t-1}(v) < \rreach^{t-1}(w) + \mymod{w-v} \leq \rreach^t(w) + \mymod{w-v} = \rreach^t(v).
\]
Since $\rreach^t(v) \leq \rreach^{t-1}(v)+1$,
we must then have $\rreach^t(w) = \rreach^{t-1}(w)$.
Therefore, the edge $\{\mymod{w + \rreach^t(w)},\;\mymod{w+\rreach^t(w)+1}\}$ is a (right) blocker for~$w$ and thus also a (right) boundary of~$v$, as $\rreach^t(w) + \mymod{w-v} = \rreach^t(v)$.
\end{proof}

We specifically note the following consequence of \cref{thm:boundary-reach-lemma}.
\begin{corollary}\label{thm:boundary-reach-corr}
Let $v$,$w$ be two vertices of~$\Cc$.
Then $w \in O_v$ if and only if
\begin{align*}
	\rreach^t(v) &< \rreach^t(w) + \mymod{w-v} \quad\text{and}\\
	\lreach^t(v) &< \lreach^t(w) + \mymod{v-w}
\end{align*}
hold for all $t \geq 0$.
\end{corollary}
\begin{proof}
Clearly both inequalities hold for $t = 0$ (unless $v = w$ in which case we are done).
By \cref{thm:boundary-reach-lemma}, neither of the two can hold with equality for any value of~$t$.
Combining this with the fact that both sides of the inequalities can only increase by~$1$ during each time step
yields the claim.
\end{proof}

The following technical lemma will be the key to solving \dVoronoi{} games.
It's precondition can be paraphrased as ``none of the two players has caught up with their opponent yet in neither direction''.

\begin{lemma}\label{thm:main-technical-lemma}
Let $v, w$ be two vertices of~$\Cc$ and $t\geq 0$. If
\begin{align*}
\mymod{v - \lreach^{t-1}(v)} &\neq \mymod{w - \lreach^{t-1}(w)} \quad\text{and} \\
\mymod{v + \rreach^{t-1}(v)} &\neq \mymod{w + \rreach^{t-1}(w)}
\end{align*}
then
\[
\Delta^t(v, w) = u_1^t(v, w) - u_1^t(w, v) = \reach^t(v) - \reach^t(w).
\]
\end{lemma}
\begin{proof}
Note first that the assumption of the lemma must also hold for all $t' < t$
(by \cref{thm:boundary-reach-lemma}, if we ever had equality then this would propagate to all subsequent times).

We now use induction over~$t$.
Clearly, the claim holds for~$t=0$.

Observe that
\begin{equation} \label{eq:expansion-1}
\begin{aligned}
u_1^t(v, w) ={}&
u_1^{t-1}(v, w) \\
&{}+ (\lreach^t(v) - \lreach^{t-1}(v))\cdot[ \mymod{v - \lreach^t(v)} \notin \Omega^t(w)] \\
&{}+ (\rreach^t(v) - \rreach^{t-1}(v))\cdot [ \mymod{v + \rreach^t(v)} \notin \Omega^t(w) ].
\end{aligned}
\end{equation}

Let now $t_{vw}$ be the first time that $v$~and~$w$ ``meet'' inside $[v, w]$, i.e., the smallest value for which
$\rreach^{t_{vw}}(v) + \lreach^{t_{vw}}(w) \geq \mymod{w-v}$.
Let analogously $t_{wv}$ be minimal with $\rreach^{t_{wv}}(w) + \lreach^{t_{wv}}(v) \geq \mymod{v-w}$.

Assume now that $\rreach^t(v) > \rreach^{t-1}(v)$. Then we claim that
\begin{equation*}
	\mymod{v + \rreach^t(v)} \in \Omega^t(w) \iff t \geq t_{vw}.
\end{equation*}
The reason for this is that our assumption from the lemma implies that $\rreach^{t-1}(w) < \mymod{v-w} + \rreach^{t-1}(v)$,
therefore $w$~cannot, at time~$t-1$, have reached $\mymod{v + \rreach^{t-1}(v)}$ by going in positive direction.
Consequently, $w$~cannot reach $\mymod{v + \rreach^{t-1}(v) + 1} = \mymod{v + \rreach^t(v)}$ until time~$t$ by going in positive direction.
This establishes that $\mymod{v + \rreach^t(v)} \in \Omega^t(w)$ implies $t \geq t_{vw}$.
Moreover, if $t\geq t_{vw}$, then by the definition of $t_{vw}$ we have that $\mymod{v + \rreach^t(v)} \in \Omega^t(w)$.

For symmetrical reasons, if $\lreach^t(v) > \lreach^{t-1}(v)$, then
\[
	\mymod{v - \lreach^t(v)} \in \Omega^t(w) \iff t \geq t_{wv}.
\]

Together, these two claims allow us to rewrite \eqref{eq:expansion-1} as follows:
\begin{equation} \label{eq:expansion}
\begin{aligned}
u_1^t(v, w) ={}&
u_1^{t-1}(v, w) \\
&{}+ (\lreach^t(v) - \lreach^{t-1}(v))\cdot[ t < t_{wv}] \\
&{}+ (\rreach^t(v) - \rreach^{t-1}(v))\cdot [ t < t_{vw} ]
\end{aligned}
\end{equation}

Our assumptions from the lemma directly imply that $\rreach^{t-1}(v) <\allowbreak \mymod{w-v} +\allowbreak \rreach^{t-1}(w)$
and $\lreach^{t-1}(v) < \mymod{v-w} + \lreach^{t-1}(w)$. 
This and the definitions of $t_{vw}$ and $t_{wv}$ give us the following two implications (since any edge that has already been crossed by the player starting at $w$ must also be available to the player starting at~$v$).
\begin{align}
t \geq t_{wv} &\implies \lreach^t(v) = \lreach^{t-1}(v) + 1 \label{eq:impl1}\\
t \geq t_{vw} &\implies \rreach^t(v) = \rreach^{t-1}(v) + 1 \label{eq:impl2}
\end{align}
By symmetry, \eqref{eq:expansion}--\eqref{eq:impl2} also hold with $v$ and $w$ swapped.
From \eqref{eq:expansion}, we obtain
\begin{align*}
u_1^t(v, w) - u_1^t(w, v) ={}&
u_1^{t-1}(v, w) - u_1^{t-1}(w, v) \\
& \left. \begin{aligned}
{}+ (\lreach^t(v) - \lreach^{t-1}(v) - \rreach^t(w) + \rreach^{t-1}(w)) \cdot [t < t_{wv}] \\
{}+ (\rreach^t(v) - \rreach^{t-1}(v) - \lreach^t(w) + \lreach^{t-1}(w)) \cdot [t < t_{vw}] 
\end{aligned} \right\} (\star)\\
\intertext{where we may use the induction hypothesis to get}
={}& \lreach^{t-1}(v) + \rreach^{t-1}(v) - \lreach^{t-1}(w) - \rreach^{t-1}(w) + (\star) \\
={}& \phantom{{}+{}} (\lreach^{t-1}(v) - \rreach^{t-1}(w)) \cdot [t \geq t_{wv}] \\
&{}+ (\rreach^{t-1}(v) - \lreach^{t-1}(w)) \cdot [t \geq t_{vw}] \\
&\left. \begin{aligned}
&{}+ (\lreach^t(v) - \rreach^t(w)) \cdot [t < t_{wv}] \\
&{}+ (\rreach^t(v) - \lreach^t(w)) \cdot [t < t_{vw}] 
\end{aligned} \right\} (\star\star) \\
\intertext{and, using~\eqref{eq:impl1} and \eqref{eq:impl2},}
={}& \phantom{{}+{}} (\lreach^t(v) - \rreach^t(w)) \cdot [ t \geq t_{wv}] \\
&{}+ (\rreach^t(v) - \lreach^t(w)) \cdot [ t \geq t_{vw}] \\
&{}+ (\star\star) \\
={}& \lreach^t(v) - \rreach^t(w)  + \rreach^t(v) - \lreach^t(w) \\
={}& \reach^t(v) - \reach^t(w).
\qedhere
\end{align*}
\end{proof}

With \cref{thm:main-technical-lemma} at hand, we can now easily determine the winner of a \dVoronoi{} game if both players play in their opponent's outer boundary component.
\begin{theorem} \label{thm:vor-cycle-diff-payout}
Let $v$ and $w$ be vertices of~$\Cc$ with $v \in O_w$ and $w \in O_v$.
Then, for all~$t \geq 0$,
\[
	u_1^t(v, w) - u_1^t(w, v) = \reach^t(v) - \reach^t(w).
\]
\end{theorem}
\begin{proof}
Since $w \in O_v$, we have that 
$\rreach^t(v) < \mymod{w-v} + \rreach^t(w)$
and
$\lreach^t(v) < \allowbreak \mymod{v-w} + \allowbreak \lreach^t(w)$
by \cref{thm:boundary-reach-corr}.
Analogously, these hold with $v$~and~$w$ swapped.
Thus we may apply \cref{thm:main-technical-lemma}.
\end{proof}

It remains to consider what happens if the players do not play in each others outer boundary component, i.e., the case not covered by \cref{thm:vor-cycle-diff-payout}.
Here and in the following, $\jointime = \jointime(v, w)$ denotes the first time at which $\Omega^\jointime(v)$ and $\Omega^\jointime(w)$ are inclusion-wise comparable.
\begin{theorem}\label{thm:vor-cycle-diff-payout-2}
Let $v$ and $w$ be vertices of $\Cc$ with $w \notin O_v$.
Let $\jointime = \jointime(v, w)$ and assume that $\Omega^\jointime(v) \supseteq \Omega^\jointime(w)$.
Then
\begin{align}
	u_1^t(v, w) - u_1^t(w, v) &= \reach^t(v) - \reach^t(w) &\forall t \leq \jointime\,, \tag{i}\label{eq:vor-cycle-diff-payout-2a}
	\\u_1^t(v, w) - u_1^t(w, v) &\geq \reach^t(v) - \reach^t(w) \geq 0 &\forall t \geq \jointime\,. \tag{ii}\label{eq:vor-cycle-diff-payout-2b}
\end{align}
\end{theorem}
\begin{proof}
The proof of \eqref{eq:vor-cycle-diff-payout-2a} is just a direct application of \cref{thm:main-technical-lemma}.

Assume without loss of generality that $\J_v(w)$ is a right boundary component.
Then $v$~has a right boundary~$B = \{b, \mymod{b+1}\}$ with $v < w < B$.
Observe that $\jointime \leq \at(v, b) < \appear{B}$ (as $\Omega^{\at(v,b)}(v)\supseteq \Omega^{\at(v,b)}(w)$)
and thus $\reach^\jointime(v) - \reach^\jointime(w) \geq 0$ follows from the definition \eqref{eq:reach} of $\reach$.
Since, $\Omega^t(v) \supseteq \Omega^t(w)$ for all $t \geq \jointime$,
we can further conclude from~\eqref{eq:reach} that $\reach^t(v) - \reach^t(w) \geq 0$ also for all $t \geq \jointime$.

So it remains to show the first inequality of~\eqref{eq:vor-cycle-diff-payout-2b}.
For this we also use induction, starting at $t = \jointime$, for which the validity follows from \eqref{eq:vor-cycle-diff-payout-2a}.
Then (for $t \geq \jointime$) we have that
$\lreach^t(v) - \lreach^t(w)$ is a nonincreasing function of~$t$
and $\rreach^t(v) - \rreach^t(w)$ is a constant function of~$t$.
Therefore, as $t \geq \jointime$ and thus $\Omega^t(v) \supseteq \Omega^t(w)$, using the induction hypothesis, we get
\begin{align*}
u_1^t(v, w) - u_1^t(w, v) ={}&
u_1^{t-1}(v, w) - u_1^{t-1}(w, v) \\
&{} + (\lreach^t(v) - \lreach^{t-1}(v)) \cdot [\mymod{v - \lreach^t(v)} \notin \Omega^t(w)]  \\
\geq{}& u_1^{t-1}(v, w) - u_1^{t-1}(w, v) \\
\geq{}& \reach^{t-1}(v) - \reach^{t-1}(w) \\
={}& \lreach^{t-1}(v) - \lreach^{t-1}(w) + \rreach^{t-1}(v) - \rreach^{t-1}(w) \\
\geq{}& \lreach^t(v) - \lreach^t(w) + \rreach^t(v) - \rreach^t(w) \\
={}& \reach^t(v) - \reach^t(w). \qedhere
\end{align*}
\end{proof}

\Cref{thm:vor-cycle-diff-payout,thm:vor-cycle-diff-payout-2} already show that a vertex~$v$ is a good choice to play on
if $\reach^t(v)$ is large for large values of~$t$.
We make this more precise in the following theorem. For this, we first need to introduce the colexicographic order:
For two monotone sequences $(a^i)_i$, $(b^i)_i$ of numbers in~$\oneto{n}$, we write
\[
	(a^i)_i > (b^i)_i \iff \exists i: a^i > b^i 
	\land \forall j > i: a^j \geq b^j.
\]
Note that, due to the monotonicity and boundedness of the sequences, this order is total.

We will use the shorthand notation $v \succ w$ for two vertices $v$, $w$ to denote that
\begin{align*}
	(\reach^t(v))_{t = \jointime(v, w)}^\infty &> (\reach^t(w))_{t = \jointime(v, w)}^\infty
\intertext{which, by \eqref{eq:reach}, is equivalent to}
	 (\abs{\Omega^t(v)})_{t = \jointime(v, w)}^\infty &> (\abs{\Omega^t(w)} )_{t = \jointime(v, w)}^\infty.
\intertext{Further, we write $v \sim w$ if}
	(\reach^t(v))_{t = \jointime(v, w)}^\infty &= (\reach^t(w))_{t = \jointime(v, w)}^\infty.
\end{align*}
Since the colexicographic order is total, we have $v \prec w$ if and only if $v \not\succsim w$.
We remark that~$\sim$~is not transitive and thus not an equivalence relation (see e.g.\ \cref{fig:intransitive}).

\begin{figure}
	\centering
	\begin{tikzpicture}[scale=1.3]
		\path[edge,radius=1]
		(180:1) node[label=180:$x$] {}
		\foreach[count=\i] \label in {1, 4, 1, 2, 4, 1}{
			arc[start angle=180+(\i-1)*60, end angle=180+\i*60]
			node[vertex, fill=white] {}
			node[at={(180+\i*60-30 : 0.8)}] {$\label$}
		}
		(-60:1) node[label=-60:$y$] {}
		(60:1) node[label=60:$z$] {}
		;
	\end{tikzpicture}
	\caption{
		In this monotonically growing cycle, $x \sim y$ and $x \sim z$, but $y \succ z$.
	}	
	\label{fig:intransitive}
\end{figure}
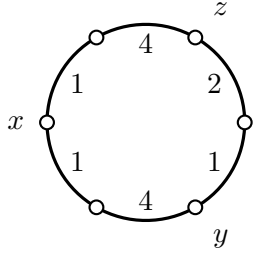

Using this newly introduced notation, we can now summarize our findings on \dVoronoi{} games as follows.
\begin{theorem}\label{thm:diffvoronoi-main}
Let $v$, $w$ be two vertices in $\mathcal{C}$ and $\jointime = \jointime(v, w)$. Then,
\[
u_1(v, w) \geq u_1(w, v) \iff \reach^{\bar{t}}(v) \geq \reach^{\bar{t}}(w)
\iff v \succsim w.
\]
\end{theorem}
\begin{proof}
The second equivalence is due to the fact that $\Omega^t(v)$ and $\Omega^t(w)$ are inclusion-wise comparable for all $t \geq \jointime$.
So it remains to prove the first equivalence.

Assume $\reach^{\bar{t}}(v) \geq \reach^{\bar{t}}(w)$,
i.e.\ $\Omega^{\bar{t}}(v) \supseteq \Omega^{\bar{t}}(w)$ by choice of~$\bar{t}$ and~\eqref{eq:reach}.

If $w \in O_v$, then $\Omega^{\bar{t}}(v) = V$ by \cref{thm:outer-property}.
This already tells us that all edges appear at or before time~$\bar{t}$ (we use $O_v \neq \emptyset$).
In particular we have for all $t \geq \bar{t}$ that
$\reach^t(v) = \reach^{t-1}(v)+2$ and $\reach^t(w) = \reach^{t-1}(w)+2$
and thus $\reach^t(v) - \reach^t(w) = \reach^{\bar{t}}(v) - \reach^{\bar{t}}(w)$.
Furthermore $v \in O_w$ by \cref{thm:outer-beats-inner}, so the claim follows by \cref{thm:vor-cycle-diff-payout}.

If $w \notin O_v$, then we can directly apply \cref{thm:vor-cycle-diff-payout-2}~\eqref{eq:vor-cycle-diff-payout-2b}.

For the other implication, assume (by symmetry) $\omega^{\bar{t}}(v) > \omega^{\bar{t}}(w)$ and thus $\Omega^{\jointime}(v) \supset \Omega^\jointime(w)$.
If~$w \in O_v$, then the proof works exactly as above.
If~$w \notin O_v$, then by \cref{thm:vor-cycle-diff-payout-2},
$u_1^{\bar{t}}(v, w) > u_1^{\bar{t}}(w, v)$.
As $u_1^t(w, v) = u_1^{\bar{t}}(w, v)$
and $u_1^t(v, w) \geq u_1^{\bar{t}}(v, w)$ for all $t \geq \bar{t}$,
this proves the claim.
\end{proof}

Note that \cref{thm:diffvoronoi-main} also holds if all inequalities are replaced by their strict versions,
since this is tantamount to simply negating and mirroring all three equivalent statements.

We call a vertex~$v$ \emph{paramount} if $v \succsim w$ for all other vertices~$w$.
Note that the existence of at least one paramount vertex~$v$ is guaranteed by simply picking~$v$ to maximize $(\omega^t(v))_{t=0}^\infty$ with respect to the colexicographic order.

We can now classify all Nash equilibria of temporal difference Voronoi games.
\begin{corollary}\label{thm:diffvoronoi-all-ne}
A pair of vertices~$(v, w)$ forms a Nash equilibrium in $\dVor(\Cc)$ if and only if $v$, $w$ are both paramount.
\end{corollary}
\begin{proof}
If $v$ and~$w$ are paramount, then for all vertices~$x$ by~\cref{thm:diffvoronoi-main} $\Delta(v, x) \geq 0$ and $\Delta(w, x) \geq 0$.
In particular, $\Delta(v, w) = -\Delta(w, v) = 0$ and neither player has a better response.

Assume now conversely that~$(v, w)$ is a Nash equilibrium.
If either vertex, say~$v$, is not paramount, then there exists a vertex~$w'$ with $w' \succ v$,
thus $\Delta(w', v) > 0$ by \cref{thm:diffvoronoi-main}.
Then also $\Delta(w, v) > 0$ since $w$~is a best response to~$v$.
Since $\Delta(w, w) = 0 > \Delta(v, w)$, $v$~cannot be a best response to~$w$.%
\footnote{
Note that we allow both players to pick the same vertex.
If we change the rules and forbid this, then \cref{thm:diffvoronoi-all-ne} still holds as long as there are at least two paramount vertices.
We leave the case of a single paramount vertex open for future research.
}
\end{proof}

We conclude this subsection by sketching an algorithm to compute all paramount vertices.
\begin{theorem}\label{thm:paramount-algo}
All paramount vertices of~$\Cc$ can be computed in~$\bigO(n^2)$~time.
\end{theorem}
\begin{proof}
For any vertex~$v$, we can compute a compact representation of the (formally infinite) vectors~$(\lreach^t(v))_{t=0}^\infty$ and~$(\rreach^t(v))_{t=0}^\infty$ in $\bigO(n)$~time,
by making note of the points in time~$t$ at which the value of~$\lreach^t(v)$ resp.\ $\rreach^t(v)$ increases.
Knowing~\eqref{eq:reach}, we need not store any further entries once we reach~$\lreach^t(v) + \rreach^t(v) \geq n$, as then clearly $t \geq \max\{\appear{e} \mid e \in E(\Cc)\}$.

Having computed these two vectors for all vertices, we can subsequently test whether any given vertex~$v$ is paramount in~$\bigO(n)$~time as follows.

Set $\ell \gets \mymod{v-1}$ and $r \gets \mymod{v+1}$.
Starting at~$t = 1$, iterate over all times~$t$ for which $\lreach^t(v) > \lreach^{t-1}(v)$ or $\rreach^t(v) > \rreach^{t-1}(v)$,
i.e., all times where~$v$ is not stuck between two boundaries.
Our algorithm will maintain the invariant that $x \precsim v$ holds for all vertices~$x$ with $\ell < x < r$.

For any such time~$t$, compare $\Omega^t(\ell)$ and~$\Omega^t(v)$.
If $\Omega^t(\ell) \supset \Omega^t(v)$, then~$v$~is clearly not paramount and we can abort.
If conversely $\Omega^t(\ell) \subseteq \Omega^t(v)$, then we must have $\ell \precsim v$ since we did not abort at any previous step.
Thus, we may then decrement~$\ell$ by~$1$ and repeat the comparison.

Symmetrically, compare also $\Omega^t(r)$ and~$\Omega^t(v)$,
possibly incrementing~$r$.

During the above, if $\ell$~and~$r$ ever pass each other (i.e., if $\mymod{r + 1} = \ell$), then we can immediately conclude that~$v$~is paramount
as we have compared~$v$ to all other vertices.

If this does not happen, then we must reach a point where $\Omega^t(v)$ is inclusion-wise incomparable to both, $\Omega^t(\ell)$ and~$\Omega^t(r)$.
Then observe that $\Omega^t(v)$ must also be incomparable to~$\Omega^t(x)$ for all vertices $x$ with $r < x < \ell$.
We then continue with the next value of~$t$.

Note that once we reach~$\lreach^t(v) + \rreach^t(v) \geq n$,
we will iterate at most~$\floor{n/2}$ more times --- afterwards every vertex has reached every other vertex.
Thus we perform at most~$\bigO(n)$ loop iterations.
Since we also update $\ell$ and~$r$ at most~$n$~times overall,
our algorithm runs in~$\bigO(n)$~time.
\end{proof}

\subsubsection{Temporal Voronoi games}
\label{sec:voronoi-cycle-normal}

Moving on to temporal Voronoi games, paramount vertices remain a reasonable choice for the players.
However, the situation becomes slightly more intricate.
Call a vertex~$w$ a \emph{paramount response} to another vertex~$v$, if $w$ is a best response to~$v$ and furthermore $w \succsim w'$ for all other best responses~$w'$.
Then in most cases a paramount vertex and a paramount response will form a Nash equilibrium;
but this does not hold universally, as seen e.g.\ in \cref{fig:paramount-counterexample}.
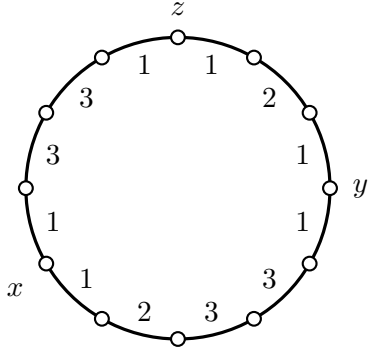
\begin{figure}
	\centering
	\begin{tikzpicture}[scale=2]
		\path[edge,radius=1]
		(210:1) node[label=219:$x$] {}
		\foreach[count=\i] \label in {1, 2, 3, 3, 1, 1, 2, 1, 1, 3, 3, 1}{
			arc[start angle=210+(\i-1)*30, end angle=210+\i*30]
			node[vertex, fill=white] {}
			node[at={(210+\i*30-15 : 0.85)}] {$\label$}
		}
		(0:1) node[label=0:$y$] {}
		(90:1) node[label=90:$z$] {}
		;
	\end{tikzpicture}
	\caption{On this temporal graph, $x$, $y$, and~$z$ are all paramount, but while $(y, z)$~form a Nash equilibrium of the temporal Voronoi game, there is no Nash equilibrium involving~$x$.}
	\label{fig:paramount-counterexample}
\end{figure}

Our goal in this section will be to prove the following theorem, which states that we will reach a Nash equilibrium by iterating the above at most two times.

\begin{theorem}\label{thm:vor-ne-cycle}
Let $x$ be a paramount vertex,
$y$ a paramount response to~$x$,
and $z$ a paramount response to~$y$.
Then $(x, y)$ or $(y, z)$ is a Nash equilibrium of $\Vor(\Cc)$.
\end{theorem}

From \Cref{thm:vor-ne-cycle}, the main result of this subsection directly follows:

\begin{corollary}
Every instance of $\Vor(\Cc)$ has a Nash equilibrium
which can be found in $\bigO(n^2)$~time.
\end{corollary}
\begin{proof}
We use \cref{thm:vor-ne-cycle}.
As seen in the previous subsection, a paramount vertex~$x$ always exists and we can find it in~$\bigO(n^2)$~time.
Since we can determine all best responses to any given vertex in $\bigO(n^2)$~time, computing~$y$ and~$z$ works similarly.
\end{proof}

The proof of \cref{thm:vor-ne-cycle} is split into three parts,
\cref{thm:vor-ne-cycle:case-in-in,thm:vor-ne-cycle:case-in-out,thm:vor-ne-cycle:case-out-out}.
For the first and easiest part, we only need the following lemma.

\begin{lemma}\label{thm:paramount-inner-boundaries}
Let $x$ be any vertex of~$\Cc$ and $y$ a best response to~$x$.
If $\J_x(y)$ is an inner boundary component, then the corresponding boundaries of~$x$ are also boundaries of~$y$.
\end{lemma}
\begin{proof}
Let $\J_x(y) = [b, b']$ be a right boundary component of~$x$ without loss of generality.
Note that $U_1(y, x) \subseteq \J_x(y)$ (\cref{thm:boundary-components})
and, since $y$~is a best response to~$x$, also $U_1(y, x) \supseteq U_1(b, x) = \J_x(y)$.
In particular, $\at(y, b) < \at(x, b) = \appear{\{\mymod{b-1},\: b\}}$
and $\at(y, b') < \at(x, b') < \appear{\{b',\: \mymod{b'+1}\}}$. 
\end{proof}

We can now give the first of three parts of the proof of \cref{thm:vor-ne-cycle} in the form of the following lemma.
It covers all cases where $z \notin O_x$.
\begin{lemma} \label{thm:vor-ne-cycle:case-in-in}
Let $x$ be a paramount vertex of~$\Cc$,
$y$ a best response to~$x$,
and $z$ a best response to~$y$
with $z \notin O_x$.
Then $(x, y)$ is a Nash equilibrium of~$\Vor(\Cc)$.
\end{lemma}
\begin{proof}
It suffices to show that $u_1(z, y) \leq u_1(x, y)$.
Let $\J_x(z) \eqqcolon [a, b]$ be a left boundary component of~$x$ without loss of generality.
We distinguish three cases.

If $y \in \J_x(z)$, then \cref{thm:paramount-inner-boundaries} applies to~$y$ (since~$\J_x(z)=\J_x(y)$ is an inner component) and yields $\J_y(z)\subseteq\J_x(z)$.
Otherwise, if $y \notin \J_x(z)$ and additionally $\at(x, b) \geq \at(y, b)$, we have $b<y<x$, as $\{\mymod{a-1},a\}$ is a left boundary of $x$ and $y\notin [a,b]$. By \Cref{thm:boundary_inheritance-cycle}, we get that $\{b,\mymod{b+1}\}$ is a left boundary of $y$, implying $\J_y(z) = \J_x(z)$ (recall \cref{thm:boundary_inheritance2-cycle}).
So we have in both of these cases that $U_1(z, y) \subseteq \J_y(z) \subseteq \J_x(z)$ (by \cref{thm:boundary-components}).
Also, $\abs{\J_x(z)} = u_1(b, x) \leq u_1(y, x)$ where the inequality holds because $y$ is a best response to $x$.
Furthermore, $u_1(y, x) \leq u_1(x, y)$ since $x$ is paramount (\cref{thm:diffvoronoi-main}).
Together this gives $u_1(z, y)\leq u_1(b,x)\leq u_1(y,x) \leq u_1(x, y)$.

It remains to consider the case that $y \notin \J_x(z)$ and $\at(x, b) < \at(y, b)$.
Since~$x$ has no blockers in~$[a, b]$ and $\at(x, a) < \appear{\{\mymod{a-1},a\}}$, we then also have $\at(x, a) <\allowbreak \at(y, a)$ and thus $\J_x(z) \subseteq U_1(x, y)$.
Further, note that $U_1(z, y) \subseteq \allowbreak U_1(z, x) \cup U_1(x, y)$.
Since~$U_1(z,x)\subseteq \J_x(z)$ (\cref{thm:boundary-components}), this yields~$U_1(z,y)\subseteq\allowbreak \J_x(z)\subseteq\allowbreak U_1(x,y)$,
which proves the claim.
\end{proof}

For the next part of the proof of \cref{thm:vor-ne-cycle}, we first need to briefly investigate how a section of~$\Cc$ (i.e., a temporal path graph)
is split between the two players if neither of them has any blockers in the area.
To this end, assume for simplicity that the temporal path graph~$\Pp$ is obtained from~$\Cc$ by deleting the edge $\{n-1, 0\}$.

\begin{lemma}\label{thm:interval-without-boundaries}
Let~$ x<v\leq w< y$ be vertices of~$\Pp$.
If $x$ and~$y$ have no blockers in~$[v,w]$, then in~$\Vor(\Pp)$
\begin{align*}
U_1(x, y) \cap [v,w] &= [v, \min\{\ceil{z}-1, w\}] \text{ and}\\
U_1(y, x) \cap [v,w] &= [\max\{\floor{z}+1, v\}, w],
\end{align*}
where
\[
z =  \frac{v + w + \at(y,w) - \at(x, v)}{2}.
\]
\end{lemma}
\begin{proof}
It is easy to verify that,
unless one of $x, y$ reaches all of $[v, w]$ before the other,
the two will meet each other exactly at point~$z$ as above,
where $z$ might be either an integer (i.e., a vertex) or have fractional part $0.5$ (i.e., an edge).
If $z$ is an integer, then the corresponding vertex is colored gray,
and $\ceil{z}-1 = z -1 \in U_1(x, y)$ and $\floor{z}+1 = z+1 \in U_1(y, x)$.
Otherwise, we have $\ceil{z}-1 = \floor{z} \in U_1(x, y)$
and $\floor{z}+1 = \ceil{z} \in U_1(y, x)$.
This proves the claim.
\end{proof}
From \Cref{thm:interval-without-boundaries}, we can directly conclude the following,
in which $\med(a, b, c)$ denotes the median of the three numbers~$a, b, c$.
\begin{lemma}\label{thm:no-boundaries-payout}
Let $x < v \leq w < y$ be vertices of~$\Pp$.
Define
\[ M \coloneqq \med\left(0, \ceil*{\frac{w-v + \at(y,w) - \at(x,v)}{2}}, w-v+1 \right). \]
Then the following holds for~$\Vor(\Pp)$:
\begin{enumerate}[(i)]
	\item If $x$ has no blockers in~$[v, w]$, then $\abs{U_1(x, y) \cap [v, w]} \geq M$.
	\item If $y$ has no blockers in~$[v, w]$, then $\abs{U_1(x, y) \cap [v, w]} \leq M$.
\end{enumerate}
\end{lemma}

We can now apply \cref{thm:interval-without-boundaries} to an outer boundary component on the cycle as follows.
\begin{lemma}\label{thm:outer-interval-response}
Let $x$, $a$, and $b$ be vertices of~$\Cc$, 
$O_x = [a, b]$
and suppose that $\at(x, a) \leq \at(x, b)$.
Then \[u_1(a, x) \geq \min\left\{\ceil*{\frac{\mymod{b - a} + 1 + \at(x,b) - \at(x, a)}{2}}, \mymod{b - a} + 1\right\}.\]
\end{lemma}
\begin{proof}
If $a = b$, then the statement is easy to verify.
Otherwise, $\{a, \mymod{a+1}\}$ is no boundary of~$x$, therefore
$\at(a, \mymod{a+1}) = \appear{\{a, \mymod{a+1}\}} \leq \at(x, a)$
and we may assume equality without loss of generality since we are interested in a lower bound for $u_1(a, x)$.
Then, as $x$ has no boundary in~$[a,b]$, $a$ has no blockers in $[\mymod{a+1}, b]$. Thus, by applying~\cref{thm:no-boundaries-payout}~(i), we obtain
\begin{align*}
u_1(a, x) &= 1 + \abs{U_1(a, x) \cap [\mymod{a+1}, b]} \\
 &\geq 1 + \med\left(0, \ceil*{\frac{\mymod{b- (a+1)} + \at(x,b) - \at(a,\mymod{a+1})}{2}},\; \mymod{b-a} \right) \\
&= \min\left\{\ceil*{\frac{\mymod{b-a}+1 + \at(x,b) - \at(x,a)}{2}},\; \mymod{b-a}+1 \right\},
\end{align*}
where the first equality holds by \cref{thm:boundary-components} as $\J_x(a)= [a,b]$ and the last equality holds as we have assumed that $\at(a,a+1)=\at(x,a)$ and $\at(x,a)\leq \at(x,b)$.
\end{proof}

We also record the following easy observation.
\begin{lemma}\label{thm:trivial-lemma1}
Let $v$ be a vertex of~$\Cc$ with left boundary~$A$ and right boundary~$B$.
If a vertex~$w \notin [A, B]$ has no boundaries inside%
\footnote{We do not consider the boundaries $A$~and~$B$ to be inside~$[A, B]$.}%
~$[A, B]$,
then $U_1(v, w) \supseteq [A, B]$.
\end{lemma}
\begin{proof}
Suppose a vertex $x \in [A, B]$ had $\at(w, x) \leq \at(v, x)$.
Say without loss of generality $x \in [A, v]$.
Since $\appear{A} > \at(v, x)$,
$w$~must first reach~$x$ via~$v$.
By \cref{thm:boundary}, this contradicts $w$ not having any boundaries in~$[A, v]$
(we can treat~$\Cc$ as a path by cutting at~$A$).
\end{proof}

Using the results above, we can now prove the second puzzle piece of~\cref{thm:vor-ne-cycle}.
\begin{lemma} \label{thm:vor-ne-cycle:case-in-out}
Let $x$ be a paramount vertex of~$\Cc$,
$y$ a best response to~$x$,
and $z$ a best response to~$y$
and suppose that $y \notin O_x$ and $z \in O_x$.
Then $(x, y)$ is a Nash equilibrium of~$\Vor(\Cc)$.
\end{lemma}
\begin{proof}
It suffices to show $u_1(x, y) \geq u_1(z, y)$.
Without loss of generality $\J_x(y)$ is a right boundary component of~$x$.
Let $B_1$ be the rightmost boundary~$B$ of~$x$ with $x < B < y$ (note that $B_1$ is one of the two boundaries enclosing~$\J_x(y)$).
Let further $B_2$ be the rightmost (right) boundary of~$x$
and $B_3$ the leftmost (left) boundary of~$x$,
that is, $[B_2, B_3] = O_x$.
See \cref{fig:vor-ne-cycle:case-in-out} for an illustration.
We set $t_i \coloneqq \appear{B_i}$.

\begin{figure}
\centering
\begin{tikzpicture}[scale=2]
\path[timeline,radius=1,every label/.style=absolute]
	(-50:1) node[vboundary,rotate=-50,label=-50:$B_1$] {}
	arc[start angle=-50, end angle=40] node[midway,right] {$\eta$}
	node[vboundary,rotate=40,label=40:$B_2$] {}
	arc[start angle=40, end angle=160] node[midway,above] {$\zeta$}
	node[vboundary,rotate=160,label=160:$B_3$] {}
	arc[start angle=160, end angle=310] node[midway,below] {$\xi$}
	;
\path
	(270:1) node[vertex,fill=white,label=90:$x$] {}
	(-30:1) node[vertex,fill=white,label=150:$y$] {}
	(120:1) node[vertex,fill=white,label=-60:$z$] {}
	;
\end{tikzpicture}
\caption{Illustration of the proof of \cref{thm:vor-ne-cycle:case-in-out}; positive direction oriented counterclockwise.}
\label{fig:vor-ne-cycle:case-in-out}
\end{figure}
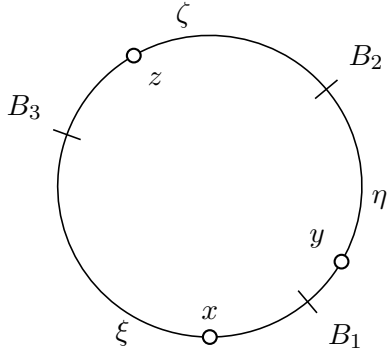

First consider the case that $x \notin O_y$.
Then $y$ has a left boundary $B$ with $B < x < y$.
By \cref{thm:boundary_inheritance-cycle}, $B$~is also a left boundary of~$x$, therefore $B_3 < B < x$.
Thus, $B_3$ is also a left boundary of~$y$ by \cref{thm:boundary_inheritance2-cycle}.
Since $B_2$, too, is a common right boundary of~$x$~and~$y$ by \cref{thm:boundary_inheritance-cycle},
we then have $\at(x, v) = \at(y, v)$ for all $v \in [B_2, B_3]$,
and in particular $O_y = O_x = [B_2, B_3]$.
From this and \cref{thm:boundary-components} we get that $U_1(z, y) = U_1(z, x)$.
Thus,
$u_1(x, y) \leq u_1(z, y) = u_1(z, x) \leq u_1(y, x)$
by the choice of~$y$ and~$z$.
Since $x$~is paramount, all of these inequalities must be equalities (by \cref{thm:diffvoronoi-main}), and the claim is proven.

In the second case, we have $x \in O_y$.
Let $\xi$ be the number of vertices in $[B_3, B_1]$,
$\eta$ the number of vertices in~$[B_1, B_2]$,
and $\zeta$ the number of vertices in~$O_x=[B_2, B_3]$
(compare \cref{fig:vor-ne-cycle:case-in-out}).
By \cref{thm:paramount-inner-boundaries}, $B_1$ is a left boundary of~$y$,
and, since $x \in O_y$, $B_1$ is the leftmost left boundary of~$y$.
(Note that there cannot be a left boundary of $y$ between $B_1$ and $x$, since this would contradict $B_1$ being a right boundary of $x$.)
Therefore, since $B_3$ is a left boundary of $x$ but not of $y$, $t_3 < t_1 + \xi$.
Also, by \cref{thm:boundary_inheritance-cycle} $B_2$ is a right boundary of $y$ and by \cref{thm:boundary_inheritance2-cycle} the rightmost right boundary of~$y$. Thus, $O_y=[B_2,B_1]$.

Further, we have
\begin{equation}
t_2 \geq t_1 + \eta \geq t_1 + u_1(y, x) \label{eq:B}
\end{equation}
where the first inequality holds as $B_1$ and $B_2$ are boundaries of~$x$ and the second inequality holds by \Cref{thm:boundary-components}.
Note that $\abs{t_2 - t_3} < \zeta$ since $O_x \neq \emptyset$.
Moreover, using that $O_y=[B_2,B_1]$, applying \cref{thm:trivial-lemma1} yields $[B_3, B_1] \subseteq U_1(x, y)$.
It remains to find out how many vertices from $[B_2, B_3]$ are colored by $x$ (where we already know that the rest will be colored gray).
To this end, we use \cref{thm:no-boundaries-payout} (mirrored), since $x$ clearly has no blockers in~$O_x=[B_2,B_3]$
and since we can ignore the fact that $y$~may reach vertices via~$B_3$ (as $x \in O_y$ will have reached them first).
This gives us
\begin{align}
u_1(x, y) &\geq \xi + \med\left(0, \ceil*{\frac{\zeta-1  + t_2 - t_3}{2}},\;  \zeta \right)
= \xi + \ceil*{\frac{\zeta -1 + t_2 - t_3}{2}} \label{eq:C}. \\
\intertext{Let $B_3\coloneqq[\mymod{b_3-1},b_3]$. Again using \Cref{thm:no-boundaries-payout}, since $U_1(z, y) \subseteq [B_2, B_1]$ (both ends being boundaries of~$y$)  and as we have already observed above that $y$ has no boundaries and thus in particular no blockers in $[B_3, B_1]$, we get}
u_1(z, y)  &\leq \zeta + \med\left(0, \ceil*{\frac{\xi -1 + t_1 - \at(z,b_3)}{2}} ,\; \xi\right) \notag
\\ &\leq \zeta + \med\left(0, \ceil*{\frac{\xi -1 + t_1 - t_3}{2}} ,\; \xi\right) \notag
\\ &\leq  \zeta + \ceil*{\frac{\xi -1 + t_1 - t_3}{2}}, \label{eq:D} \\
\intertext{where we used for the second inequality that $\at(z,b_3)\geq t_3$. Finally, by \cref{thm:outer-interval-response} and as $y$ is a best response to $x$, it follows
}
u_1(y, x) &\geq \min\left\{\ceil*{\frac{\zeta + \abs{t_2 - t_3}}{2}}, \zeta\right\}
= \ceil*{\frac{\zeta  + \abs{t_2 - t_3}}{2}}. \label{eq:E}
\end{align}

Using the above, we then derive by \eqref{eq:C} and \eqref{eq:D}
\begin{align*}
2(u_1(x, y) - u_1(z, y)) &\geq  
\mathrlap{ 2\xi + 2\ceil*{\frac{\zeta -1 + t_2 - t_3}{2}}
-2\zeta - 2\ceil*{\frac{\xi -1 + t_1 - t_3}{2}} }
\\ &\geq \mathrlap{2\xi - 2\zeta + \zeta -1 + t_2 - t_3 - \xi+1 - t_1 + t_3 -1}
\\ &= \xi - \zeta + t_2 - t_1 - 1
\\ &\geq t_3 + t_2 - 2t_1 - \zeta  \quad&\text{since $t_3 < t_1 + \xi$}
\\ &\geq t_3 + t_2- 2t_1 - 2 u_1(y,x) +\abs{t_2-t_3} \quad&\text{by \eqref{eq:E}}
\\ &\geq t_3 - t_2 + \abs{t_2-t_3} \quad&\text{by \eqref{eq:B}}
\\ &\geq 0
\end{align*}
Thus, $u_1(x, y) \geq u_1(z, y)$.
\end{proof}

The following easy result will help us in the last part of the proof of \cref{thm:vor-ne-cycle}.

\begin{lemma}\label{thm:outer-gray}
If $v \in O_w$ and $w \in O_v$ are two vertices of~$\Cc$,
then the pair of positions~$(v,w)$ produces at most two gray vertices in $\Vor(\Cc)$.
\end{lemma}
\begin{proof}
By \cref{thm:outer-property},  $\Omega^t(v)\subseteq \Omega^t(w)$ or $\Omega^t(w)\subseteq \Omega^t(v)$ for some $t$ implies that $\Omega^t(v)=\Omega^t(w)=V$.
Thus, any gray vertex must be reached from the two players from opposite sides,
so there can be at most one gray vertex in each of the two components in which $v$ and $w$ divide $\Cc$.
\end{proof}

We need one last ingredient before we can solve the final piece of the puzzle:
\begin{lemma}\label{thm:rescue-lemma}
Let $x$~be a paramount vertex of~$\Cc$,
and $y$~a paramount response to~$x$ with $y \sim x$.
Then $y$~is paramount.
\end{lemma}
\begin{proof}
Suppose for contradiction that there was a vertex~$y' \succ y$.
Since $x \sim y$, clearly $y' \neq x$.
Assume without loss of generality that $y < y' < x$.
Note that we must have $\jointime(y, y') < \jointime(x, y)$,
for otherwise $y' \succ x$ would follow from the definition of~$\succ$,
contradicting the fact that~$x$~is paramount.
Therefore $\Omega^{\jointime(y, y')}(y') \neq V$,
so $y \notin O_{y'}$ by \cref{thm:outer-property}.
Thus, $\J_{y'}(y)$ is a left boundary interval. Let~$B$ be its left boundary, i.e., the rightmost left boundary of~$y'$ with $B < y < y'$.

As $\jointime(y, y') < \jointime(x, y)$, vertex~$x$ has no boundaries inside $[B, y']$, thus
$[B, y'] \subseteq \allowbreak U_1(y', x)$ by \cref{thm:trivial-lemma1}.
Then it is easy to see that $U_1(y, x) \subseteq U_1(y', x)$, i.e., $y'$~is a best response to~$x$.
Since $y'$~is paramount, this contradicts the choice of~$y$ as paramount response to~$x$.
\end{proof}

Now, we can finally prove the remaining case of \cref{thm:vor-ne-cycle}:
\begin{lemma} \label{thm:vor-ne-cycle:case-out-out}
Let $x$ be a paramount vertex of~$\Cc$,
$y$ a paramount response to~$x$,
and $z$ a paramount response to~$y$,
and suppose that $y, z \in O_x$.
Then $(x, y)$ or~$(y, z)$ is a Nash equilibrium of~$\Vor(\Cc)$.
\end{lemma}
\begin{proof}
Without loss of generality, assume that $x < y < z$.
Further, we have $x \in O_y \cap O_z$ by \cref{thm:outer-beats-inner} and \cref{thm:diffvoronoi-main}, as $x$ is paramount. 
We prove the statement via case distinction.  \medskip

\case{1}{$\mathbf{z \notin O_y}$}
Then $\J_y(z)$ is an inner boundary component of~$y$
and $U_1(z, y) \subseteq \J_y(z)\subseteq V\setminus O_y \subseteq U_1(y, x)$ by \cref{thm:boundary-components,thm:mutual-outer-inner}.
Since~$x$ is paramount, we have $u_1(y, x) \leq u_1(x, y)$ (\cref{thm:diffvoronoi-main}).
Together, this gives $u_1(z, y) \leq u_1(x, y)$.
Thus, as $z$ is a best response to $y$, $x$ is a best response to~$y$ and $(x, y)$ is a Nash equilibrium. \medskip

\case{2}{$\mathbf{z \in O_y}$ and $\mathbf{y \notin O_z}$}
Then $\J_z(y)$ is an inner boundary component of~$z$
and $U_1(y, z) \subseteq \J_z(y) \subseteq V\setminus O_z \subseteq U_1(z, x)$ by \cref{thm:boundary-components,thm:mutual-outer-inner}.
Therefore,
\[
	U_1(y, x) \subseteq U_1(y, z) \cup U_1(z, x)
	\subseteq \J_z(y) \cup U_1(z, x)
	\subseteq U_1(z, x),
\]
which means that $z$ is a best response to~$x$.
Also, by \cref{thm:outer-beats-inner}, there is a time~$t$ with $\Omega^t(z) \supset \Omega^t(y)$.
By \cref{thm:diffvoronoi-main}, this contradicts $y$~being a paramount response to~$x$.  \medskip

\case{3}{$\mathbf{z \in O_y}$ and $\mathbf{y \in O_z}$}
Define $\Delta_+ \coloneqq U_1(z, y) \setminus U_1(x, y)$
and $\Delta_- \coloneqq \allowbreak U_1(x, y) \setminus U_1(z, y)$.
We assume $\abs{\Delta_+} > \abs{\Delta_-}$ since otherwise~$(x,y)$ is a Nash equilibrium and we are done.

Let $t_{yx}$ be the first time for which $\rreach^{t_{yx}}(y) + \lreach^{t_{yx}}(x) \geq \mymod{x-y}$,
that is, the first time that $x$ and $y$ reach a common vertex in $[y, x]$.
Let also $t_{yz}$ be the first time for which $\rreach^{t_{yz}}(y)+ \lreach^{t_{yz}}(z)  \geq \mymod{z-y}$.
Note that as $x<y<z$, we have $t_{yz}\leq t_{yx}$.

Then we claim that $t_{yz} < t_{yx}$.
Otherwise, if $t' \coloneqq t_{yz} = t_{yx}$, then
we would have
\[
	\mymod{x-y} - \lreach^{t'}(x) = 
	\rreach^{t'}(y) = 
	\mymod{z-y} - \lreach^{t'}(z)
\]
and thus
\[
	\lreach^{t'}(x) = \lreach^{t'}(z) + \mymod{x-y} - \mymod{z-y}
	= \lreach^{t'}(z) + \mymod{x-z}
\]
in contradiction to $z \in O_x$ (\cref{thm:boundary-reach-corr}).

Since $z \in O_y$ and $y \in O_z$, \cref{thm:boundary_inheritance-cycle} implies that
\begin{align}
\forall t \geq t_{yz}:& \quad  \lreach^t(z) = \lreach^{t-1}(z) + 1 \text{ and }  \rreach^t(y) = \rreach^{t-1}(y) + 1  \label{eq:Z}
\intertext{and similarly, as $y\in O_x$}
\forall t \geq t_{yx}:& \quad \lreach^t(x) = \lreach^{t-1}(x) + 1. \label{eq:Z2}
\end{align}

Let $\delta_+ \coloneqq t_{yx} - t_{yz} > 0$ (where the inequality holds as we have shown above that $t_{yz}<t_{yx}$) and define $\lambda_{yx}$, $\lambda_{yz} \in \{0, 1\}$ as follows:
\begin{align*}
\rreach^{t_{yx}}(y) + \lreach^{t_{yx}}(x) &= \mymod{x-y} + \lambda_{yx},\\
\rreach^{t_{yz}}(y)+ \lreach^{t_{yz}}(z) &= \mymod{z-y}+ \lambda_{yz}.
\end{align*}
Observe that $\lambda_{yx} = 0$ if and only if a player at position~$y$ and a player at~$x$ simultaneously reach some vertex in~$[y, x]$,
i.e., if some vertex in~$[y, x]$ is colored gray.
An analogous statement holds for~$\lambda_{yz}$.

Then, for all $t \geq t_{yx}>t_{yz}$, it holds
\begin{align*}
\lreach^t(x) - \lreach^t(z)
&= \lreach^{t_{yx}}(x) - \lreach^{t_{yx}}(z) &\text{by \eqref{eq:Z}, \eqref{eq:Z2}}
\\ &= \mymod{x-y}+\lambda_{yx} - \rreach^{t_{yx}}(y) - \lreach^{t_{yx}}(z)
\\ &= \mymod{x-y}+\lambda_{yx} - \rreach^{t_{yz}}(y) - \lreach^{t_{yz}}(z) - 2\delta_+ \quad &\text{by \eqref{eq:Z}}
\\ &= \mymod{x-y}+\lambda_{yx} - \mymod{z-y} - \lambda_{yz} - 2\delta_+
\\ &= \mymod{x-z} + \lambda_{yx} - \lambda_{yz} - 2\delta_+
\\ &= \mymod{x-z} - 2 \abs{\Delta_+} - \lambda_{yx} + \lambda_{yz}.
\end{align*}
In the above equation we used the fact that $\abs{\Delta_+} = \delta_+ - \lambda_{yx} + \lambda_{yz}$,
which is due to
$y$~reaching exactly one vertex from $\Delta_+$
during each time $t' \in [t_{yz} + 1 - \lambda_{yz},\allowbreak t_{yx} - \lambda_{yx}]$
(see also \cref{fig:caseIII}).

\begin{figure}[t]
	\centering
	\begin{tikzpicture}[xscale=0.5, yscale=0.5]
		\def\timelabelx{-3}
		\def\firstvertex{0}
		\gamegrid{15}{9}
		\def \lifetime {9}
		\def \n {15}
			
		\def \y {3}
		\def \z {7}
		\def \x {14}
			
		\def \colorX {\colorP}
		\def \colorY {\colorQ}
		\def \colorZ {green!70}
		
		\path (18,0); 
			
		\begin{scope}[on background layer]
			\draw[edge]
				\foreach \t in {1,2} {
					(v-1-\t) ++(-0.5,0) -- (v-\y-\t)
				}
				(v-\z-1) -- ++(1,0)
				\foreach \t in {1,2,3} {
					(v-\x-\t) ++(-2,0) -- (v-\n-\t) -- ++(0.5, 0)
				}
				(v-5-2) -- (v-\z-2) -- ++(1,0)
				(v-1-3) ++(-0.5,0) -- (v-10-3)
				\foreach \t in {4,...,\lifetime} {
					(v-1-\t) ++(-0.5,0) -- (v-\n-\t) -- ++(0.5,0)
				}
			;
		\end{scope}		
	
		\path[every node/.style={anchor=base}]
			(0, \vertexlabely)
			+(\x, 0) node {$x$}
			+(\y, 0) node {$y$}
			+(\z, 0) node {$z$}
		;
			
		\path[every node/.style={anchor=east}]
			(0, 4.95) node {$t_{yz} ={}$}
			(0, 7.95) node {$t_{yx} ={}$}
		;
			
		\draw[decorate,decoration={brace,amplitude=2mm}]
			($(v-\z-0) + (-2.1, 0.5)$) -- ($(v-\z-0) + (1.1, 0.5)$)
			(v-\z-0) ++(-0.5, 1.5) node {$\Delta_+$}
		;
			
		\begin{scope}[yshift=10 cm,yscale=-1] 
			\draw[reach=\colorY] (\y, 1) -- (0.5, \y+0.5) -- (0.5, \lifetime+1.5) -- (\y+\lifetime-1.5, \lifetime+1.5) -- (\y, 3) -- cycle;
			\draw[reach=\colorZ] (\z, 1) -- (\z, 2) -- (0.5, \z+1.5) -- (0.5, \lifetime+1.5) -- (\n+0.5, \lifetime+1.5) -- (\z+1, 3) -- (\z+1, 2) -- cycle;
			\draw[reach=\colorX] (\x, 1) -- (\x-2, 3) -- (\x-2, 4) -- (\x - \lifetime+0.5, \lifetime+1.5) -- (\n+0.5, \lifetime+1.5) -- (\n+0.5, \n+0.5 - \x + 1) -- cycle;
		\end{scope}

	\end{tikzpicture}
	\caption{
		Illustration of the proof of \cref{thm:vor-ne-cycle:case-out-out},~case~3.
		The figure shows a part of the graph~$\Cc$ and how $\Omega^t(y)$~(red), $\Omega^t(z)$~(green), and $\Omega^t(x)$~(blue) grow over time.
		Note that in this example $\lambda_{yx} = 0$ and $\lambda_{yz} = 1$
		and thus
		$\abs{\Delta^+} = 4 = t_{yx} - t_{yz} -\lambda_{yx} + \lambda_{yz}$.
	}
	\label{fig:caseIII}
\end{figure}
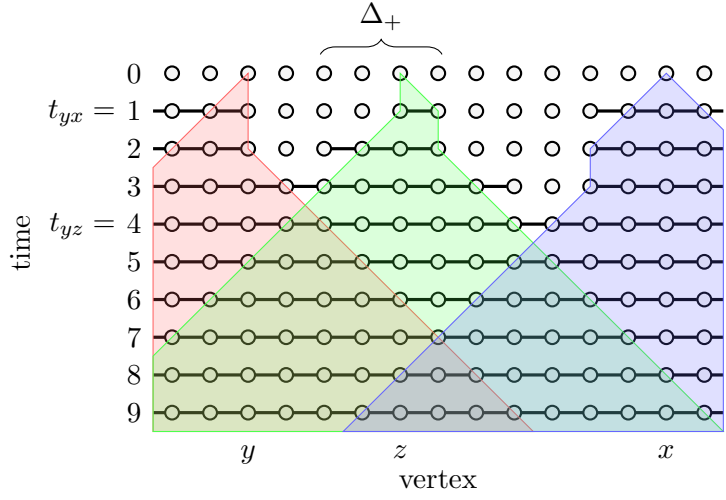

Analogously, let~$t_{zy}$ be the first time where $\rreach^{t_{zy}}(z) + \lreach^{t_{zy}}(y) \geq \mymod{y-z}$
and let~$t_{xy}$ be the first time where $\rreach^{t_{xy}}(x) + \lreach^{t_{xy}}(y) \geq \mymod{y-x}$.
Let $\delta_- \coloneqq t_{zy} - t_{xy} > 0$
and define $\lambda_{zy}$, $\lambda_{xy}$ as
\begin{align*}
\rreach^{t_{zy}}(z) + \lreach^{t_{zy}}(y) &= \mymod{y-z} + \lambda_{zy},\\
\rreach^{t_{xy}}(x)+ \lreach^{t_{xy}}(y) &= \mymod{y-x}+ \lambda_{xy}.
\end{align*}
Analogously to above, for all $t \geq t_{zy}$, it holds
\begin{align*}
\rreach^t(x) - \rreach^t(z) 
&= \rreach^{t_{zy}}(x) - \rreach^{t_{zy}}(z)
\\ &= \rreach^{t_{zy}}(x) + \lreach^{t_{zy}}(y) -\mymod{y-z} - \lambda_{zy} 
\\ &= \rreach^{t_{xy}}(x) + \lreach^{t_{xy}}(y) + 2\delta_- -\mymod{y-z} - \lambda_{zy} 
\\ &= \mymod{y-x} + \lambda_{xy} + 2\delta_- - \mymod{y-z} - \lambda_{zy}
\\ &= -\mymod{x-z} + \lambda_{xy} + 2\delta_- - \lambda_{zy}
\\ &= -\mymod{x-z} + 2 \abs{\Delta_-} - \lambda_{xy} + \lambda_{zy}.
\end{align*}
Since~$x$ is paramount, we now obtain from \cref{thm:diffvoronoi-main} that, for sufficiently large~$t$, it holds
\begin{align*}
0 &\leq 
\lreach^t(x) + \rreach^t(x) - (\lreach^t(z) + \rreach^t(z))
\\ &= \underbrace{2(\abs{\Delta_-} - \abs{\Delta_+})}_{\le -2} + \underbrace{\lambda_{yz} + \lambda_{zy} - \lambda_{xy} - \lambda_{yx}}_{\le 2}
\leq 0,
\end{align*}
from which we conclude that
\begin{align}
\reach^t(x)  &= \reach^t(z), \label{eq:x-eq-z}
\\ \abs{\Delta_+} &= \abs{\Delta_-} + 1, \nonumber
\\ \lambda_{yz} &= \lambda_{zy} = 1, \text{ and} \label{eq:lambda1}
\\ \lambda_{xy} &= \lambda_{yx} = 0. \label{eq:lambda0}
\end{align}

Furthermore,
\[
	u_1(y, x) \leq u_1(x, y) = n - 2 - u_1(y, x) \leq n-2 - u_1(z, x) \leq u_1(x, z) = u_1(z, x),
\]
where the first inequality is due to~$x$~being paramount,
the first equality is by \cref{thm:outer-gray} and \eqref{eq:lambda0},
the second inequality is due to the fact that $y$~is a best response to~$x$,
the third inequality is again by~\cref{thm:outer-gray},
and the last equality is by~\cref{thm:vor-cycle-diff-payout} and \eqref{eq:x-eq-z}.

Since also $u_1(z, x) \leq u_1(y, x)$ by choice of~$y$, all of the above inequalities are in fact equalities.
In particular, $u_1(x, y) = u_1(y, x) = (n-2)/2$.
We can thus deduce by \cref{thm:rescue-lemma} that $y$~is paramount.
Therefore also $u_1(y, z) \geq u_1(z, y)$.

We claim that~$z$ is paramount, too.
Otherwise, we must have $u_1(z, y) > u_1(x, y)$ by choice of~$z$.
Furthermore $u_1(y, z) > u_1(z, y)$ or we could deduce the claim from~\cref{thm:rescue-lemma}.
Thus we get $u_1(y, z) > u_1(z, y) > u_1(x, y) = n/2 - 1$, i.e.,
$u_1(y, z) + u_1(z, y) \geq n + 1$, which is clearly impossible. 

Due to this, $(y, z)$ is a Nash equilibrium in $\dVor(\Cc,2)$ by \cref{thm:diffvoronoi-all-ne} (as well as $(x, y)$ and~$(x, z)$).
As there are no gray vertices~(by \eqref{eq:lambda1}),
this implies that $(y, z)$ is also a Nash equilibrium in $\Vor(\Cc)$:
Any opportunity for one player to improve would have to come with a loss for the other player,
which would contradict~$(y, z)$ being a Nash equilibrium in $\dVor(\Cc, 2)$.
\end{proof}

As \cref{thm:vor-ne-cycle:case-in-in,thm:vor-ne-cycle:case-in-out,thm:vor-ne-cycle:case-out-out} together cover all possible cases,
we have now successfully proved \cref{thm:vor-ne-cycle}.
Note that in the only case where $(x, y)$~did not form a Nash equilibrium (case~3 of~\cref{thm:vor-ne-cycle:case-out-out}),
we showed~$y$~to be paramount.
Thus there always exists some paramount vertex that is part of a Nash equilibrium.

\section{Conclusion}
Our work is meant to initiate further systematic studies of 
(not only competitive) games on (classes of) temporal graphs. 
There is a wealth of unexplored research directions to pursue.

Extending our work on temporal trees and cycles, there are also many more special temporal graphs to study such as temporal grids.
Another direction is to consider variations of temporal diffusion and Voronoi 
games.
For example, as already partly studied in difference diffusion/Voronoi games, the payoff could be defined as the difference of the number of vertices 
colored by the players.
Related to this, ``splitting'' gray vertices between players is another possibility and was already studied on static games.
Finally, for Voronoi games, there are several different temporal 
distance notions to consider.
For example, one may study non-strict walks (i.e,. walks that traverse multiple edges in a single time step).
Note that these are trivial on monotonically shrinking temporal graphs.
We conjecture that \cref{thm:voronoi-tree} still holds for such non-strict Voronoi games
--- in fact, we know of no instance of any monotonically growing temporal graph without a Nash equilibrium.

\end{document}